\documentclass[11pt,a4paper]{article}

\usepackage[utf8]{inputenc}
\usepackage[T1]{fontenc}
\usepackage{lmodern}
\usepackage{amsmath,amsthm,amssymb,amsfonts,mathtools}
\usepackage{geometry}
\usepackage{setspace}
\usepackage{fancyhdr}
\usepackage{hyperref}
\usepackage{enumitem}
\usepackage{float}
\usepackage{tikz}
\usepackage{pgfplots}
\pgfplotsset{compat=1.18}
\usetikzlibrary{arrows.meta,patterns,decorations.markings,calc,positioning,shapes.geometric,decorations.pathreplacing,backgrounds,shadings,fit}
\usepackage[ruled,vlined,linesnumbered]{algorithm2e}
\usepackage{xcolor}
\usepackage{tcolorbox}
\tcbuselibrary{breakable,skins}
\usepackage{caption}
\usepackage{titlesec}

\definecolor{accent}{HTML}{2B4C7E}      % deep blue
\definecolor{accentlight}{HTML}{E8EEF6}  % very light blue
\definecolor{thmcolor}{HTML}{1B6B4A}     % dark teal (theorems)
\definecolor{defcolor}{HTML}{7B3F00}     % warm brown (definitions)
\definecolor{excolor}{HTML}{4A4A8A}      % muted indigo (examples)
\definecolor{remcolor}{HTML}{606060}      % neutral gray (remarks)

\titleformat{\section}
  {\Large\bfseries\color{accent}}
  {\thesection}{0.8em}{}
  [\vspace{-0.6em}{\color{accent}\rule{\textwidth}{0.6pt}}]
\titleformat{\subsection}
  {\large\bfseries\color{accent!80!black}}
  {\thesubsection}{0.7em}{}

\newcommand{\repourl}{https://github.com/gabayae/expansive-homeomorphisms-complexity-qmetric}

\fancypagestyle{titlepage}{%
  \fancyhf{}%
}

\newtcolorbox{insightbox}[1][]{
  colback=accent!4,
  colframe=accent,
  fonttitle=\bfseries,
  title={#1},
  boxrule=0.5pt,
  arc=3pt,
  breakable,
  left=6pt, right=6pt,
}

\theoremstyle{plain}
\newtheorem{theoreminner}{Theorem}[section]
\newtheorem{lemmainner}[theoreminner]{Lemma}
\newtheorem{propositioninner}[theoreminner]{Proposition}
\newtheorem{corollaryinner}[theoreminner]{Corollary}
\theoremstyle{definition}
\newtheorem{definitioninner}[theoreminner]{Definition}
\newtheorem{exampleinner}[theoreminner]{Example}
\newtheorem{remarkinner}[theoreminner]{Remark}

\tcolorboxenvironment{theoreminner}{
  blanker, breakable,
  left=4pt, borderline west={2.5pt}{0pt}{thmcolor},
  before skip=8pt, after skip=8pt,
}
\tcolorboxenvironment{lemmainner}{
  blanker, breakable,
  left=4pt, borderline west={2.5pt}{0pt}{thmcolor!70},
  before skip=8pt, after skip=8pt,
}
\tcolorboxenvironment{propositioninner}{
  blanker, breakable,
  left=4pt, borderline west={2.5pt}{0pt}{thmcolor!70},
  before skip=8pt, after skip=8pt,
}
\tcolorboxenvironment{corollaryinner}{
  blanker, breakable,
  left=4pt, borderline west={2.5pt}{0pt}{thmcolor!50},
  before skip=8pt, after skip=8pt,
}
\tcolorboxenvironment{definitioninner}{
  blanker, breakable,
  left=4pt, borderline west={2.5pt}{0pt}{defcolor},
  before skip=8pt, after skip=8pt,
}
\tcolorboxenvironment{exampleinner}{
  blanker, breakable,
  left=4pt, borderline west={2pt}{0pt}{excolor!60},
  before skip=8pt, after skip=8pt,
}
\tcolorboxenvironment{remarkinner}{
  blanker, breakable,
  left=4pt, borderline west={1.5pt}{0pt}{remcolor!40},
  before skip=6pt, after skip=6pt,
}

\tcolorboxenvironment{proof}{
  colback=black!3,
  colframe=black!15,
  boxrule=0.4pt,
  arc=2pt,
  breakable,
  left=6pt, right=6pt, top=4pt, bottom=4pt,
  before skip=8pt, after skip=8pt,
}

\newenvironment{theorem}[1][]{\begin{theoreminner}[#1]}{\end{theoreminner}}
\newenvironment{lemma}[1][]{\begin{lemmainner}[#1]}{\end{lemmainner}}
\newenvironment{proposition}[1][]{\begin{propositioninner}[#1]}{\end{propositioninner}}
\newenvironment{corollary}[1][]{\begin{corollaryinner}[#1]}{\end{corollaryinner}}
\newenvironment{definition}[1][]{\begin{definitioninner}[#1]}{\end{definitioninner}}
\newenvironment{example}[1][]{\begin{exampleinner}[#1]}{\end{exampleinner}}
\newenvironment{remark}[1][]{\begin{remarkinner}[#1]}{\end{remarkinner}}

\newcommand{\R}{\mathbb{R}}
\newcommand{\N}{\mathbb{N}}
\newcommand{\Z}{\mathbb{Z}}
\newcommand{\C}{\mathcal{C}}
\newcommand{\BigO}{\mathcal{O}}
\newcommand{\eps}{\varepsilon}

\hypersetup{
    colorlinks=true,
    linkcolor=accent,
    citecolor=thmcolor,
    urlcolor=excolor,
    hypertexnames=false
}

\begin{document}

% ═══════════════════════════════════════════════════════════════
%  TITLE PAGE
% ═══════════════════════════════════════════════════════════════
\pagenumbering{Alph}
\begin{titlepage}
\thispagestyle{titlepage}
\centering
\vspace*{0.6cm}

{\color{accent}\rule{0.4\textwidth}{0.8pt}}

\vspace{0.8cm}

{\huge\bfseries\color{accent} Expansive homeomorphisms on\\[0.5em]
 complexity quasi-metric spaces}

\vspace{0.5cm}
{\Large\itshape\color{accent!70!black} A bridge between dynamical systems\\[0.2em]
 and computational complexity theory}

\vspace{0.3cm}
{\color{accent}\rule{0.25\textwidth}{0.5pt}}

\vspace{0.8cm}

{\large
\textbf{Ya\'e U.\ Gaba}$^{\dagger,\ddagger,\S}$
}

\vspace{0.3cm}

{\small\color{red!70!black}
\href{https://orcid.org/0000-0001-8128-9704}{\textsc{orcid}}%
\enspace$\cdot$\enspace
\href{https://arxiv.org/a/gaba_y_1.html}{\textsc{arXiv}}%
\enspace$\cdot$\enspace
\href{https://scholar.google.com/citations?user=UTszjV4AAAAJ&hl=en}{\textsc{Google Scholar}}%
}

\vspace{0.5cm}

{\small
$^{\dagger}$\,AI Research and Innovation Nexus for Africa (AIRINA Labs),
AI.Technipreneurs, B\'enin\\[0.3em]
$^{\ddagger}$\,Sefako Makgatho Health Sciences University (SMU),
South Africa\\[0.3em]
$^{\S}$\,African Center for Advanced Studies (ACAS), Cameroon
}

\vspace{1.5cm}
\begin{abstract}
\noindent
The complexity quasi-metric of Schellekens is a topological
framework in which the asymmetry of computational comparisons
---``$A$ is at most as fast as $B$'' carrying different information
than ``$B$ is at most as slow as $A$''---is built into the distance
itself.  This paper develops the theory of expansive homeomorphisms
on the resulting space.  The central result is that the scaling
transformation $\psi_\alpha(f)(n)=\alpha f(n)$ is expansive on the
complexity space $(\C,d_\C)$ if and only if $\alpha\neq 1$.
The $\delta$-stable sets coincide with closed
$d_\C$-balls and contain, in particular, every function
pointwise dominated by the basepoint.  We further establish the
exact identity $d_\C^s(\psi_\alpha^n(f),\psi_\alpha^n(g))
=\alpha^{-n}\,d_\C^s(f,g)$ for all $n\ge 0$ (so $\psi_\alpha$ acts
as a strict $(1/\alpha)$-contraction on $d_\C^s$ for $\alpha>1$),
and we connect orbit separation to the classical time hierarchy
theorem of Hartmanis and Stearns.  Unstable sets, conjugate
dynamics, and a no-compact-invariant-set result---showing that
the standard Bowen entropy on compact invariant subsets of
$(\C,d_\C^s)$ is necessarily trivial---are also worked out.  Concrete algorithms and Python
implementations accompany every proof, so each result can be
checked computationally; SageMath snippets sit alongside the
examples, and the full code is in the
\href{\repourl}{companion repository}.

\medskip
\textbf{Keywords:} Expansive homeomorphism; complexity quasi-metric;
computational complexity; asymmetric topology; dynamical systems;
stable and unstable sets

\smallskip
\textbf{2020 MSC:} 54H20; 37B20; 68Q25; 54E15; 54E55
\end{abstract}
\vfill
\end{titlepage}
\pagenumbering{arabic}

\vspace*{0.5cm}
\begin{flushright}
\begin{minipage}{0.72\textwidth}
\raggedleft\itshape\color{accent!80!black}
``The study of iteration, the study of the behaviour of a transformation
when it is repeated, is the fundamental problem of dynamics.''
\vspace{0.3em}

\upshape\color{accent!60!black}--- Henri Poincar\'e
\end{minipage}
\end{flushright}
\vspace{0.8cm}

% ═══════════════════════════════════════════════════════════════
\section{Introduction: where dynamics meets computation}
\label{sec:intro}
% ═══════════════════════════════════════════════════════════════

\subsection{A tale of two theories}

This paper builds a bridge between two subjects that rarely meet:
\emph{dynamical systems}, which studies how a system evolves under
iteration, and \emph{computational complexity theory}, which studies
the resources required to solve computational problems.

At first glance these two fields appear to have little in common.
Dynamical systems theory asks: given a map $\psi\colon X\to X$,
how do the orbits $\{x,\psi(x),\psi^2(x),\ldots\}$ behave as the
number of iterations grows?  Do nearby orbits stay close, or do they
diverge?  If they diverge, how quickly?  Complexity theory, on the
other hand, asks: given a computational problem, how do the resources
required to solve it---time, memory, communication---grow with the
input size~$n$?

What links the two worlds is that complexity comparisons are
asymmetric.  Saying ``algorithm~$A$ is at most as fast as
algorithm~$B$'' does not say the same thing as ``algorithm~$B$
is at most as slow as algorithm~$A$.''  If $f(n)\le g(n)$ for
all~$n$ (so that $f$ is at least as fast as~$g$), the ``cost'' of
moving from $f$ to~$g$ is zero---going to a slower algorithm is
easy to simulate---but the cost of moving from $g$ to~$f$ is
positive, because shaving time off a running algorithm takes real
insight.

This asymmetry is what a \emph{quasi-metric} captures: a distance
function $q$ where $q(x,y)$ need not equal $q(y,x)$.  Schellekens~\cite{schellekens1995} introduced the
\emph{complexity quasi-metric} $d_\C$ on the space of functions
$f\colon\N\to[1,\infty)$ and showed that its topological properties
encode fundamental features of computational complexity.  Romaguera
and Schellekens~\cite{romaguera1999} subsequently developed the
quasi-metric structure of complexity spaces in depth.

On the dynamical side, \emph{expansive homeomorphisms}---maps under
which every pair of distinct points is eventually separated by more
than some fixed threshold~$\delta$---are a central object of study,
going back to Utz~\cite{utz1950} and developed extensively by
Bowen~\cite{bowen1975} and Reddy~\cite{reddy1983}.  Recently, Olela
Otafudu, Matladi, and Zweni~\cite{olela2024} extended the theory of
expansive homeomorphisms to quasi-metric spaces, opening the door to
applications in asymmetric settings.

Our contribution is to walk through that door: to apply the theory
of expansive homeomorphisms to the complexity quasi-metric space.
Basic dynamical concepts---orbits, stable sets, hyperbolicity---turn
out to have direct counterparts in complexity theory.

\paragraph{Related work.}
The complexity quasi-metric space was introduced by
Schellekens~\cite{schellekens1995}, who established its basic
topological properties and connections to denotational semantics.
Romaguera and Schellekens~\cite{romaguera1999} subsequently
developed the quasi-metric structure in depth, proving
completeness results and studying the Smyth completion.  On
the dynamical side, expansive homeomorphisms on metric spaces
have a long history, beginning with Utz~\cite{utz1950} and
substantially advanced by Bowen~\cite{bowen1975}, who connected
expansiveness to topological entropy, and Reddy~\cite{reddy1983},
who established canonical coordinate systems for expansive maps.
The extension of expansive homeomorphisms to quasi-metric spaces
was carried out by Olela-Otafudu, Matladi, and
Zweni~\cite{olela2024}, who proved that $q$-expansiveness and
$q^t$-expansiveness are equivalent and developed the abstract
theory of stable and unstable sets in the asymmetric setting.
K\"unzi~\cite{kunzi1995,kunzi2001} provided a comprehensive
account of the topology of quasi-metric spaces, which underpins
the present work.

What is new here is the \emph{combination}: \cite{olela2024}
develops the abstract theory without any specific quasi-metric
space in mind, while \cite{schellekens1995,romaguera1999} study
the complexity space without dynamical-systems tools.  Bringing
the two together gives dynamical concepts a computational meaning
and gives complexity-theoretic distinctions a dynamical form.

\begin{figure}[H]
\centering
\begin{tikzpicture}[scale=0.95]
    % Dynamics box
    \begin{scope}[shift={(-5.2,0)}]
        \draw[thick, fill=blue!6, rounded corners=18pt]
             (-3.5,-3.8) rectangle (3.5,3.8);
        \node[font=\large\bfseries, blue!70!black] at (0,3.0) {Dynamics};
        \draw[thick, blue!60, ->, >=stealth]
             plot[smooth, tension=0.7]
             coordinates {(-2.4,1.8) (-1.0,1.0) (0.8,1.2) (2.4,0.7)};
        \draw[thick, red!60, ->, >=stealth]
             plot[smooth, tension=0.7]
             coordinates {(-2.4,1.4) (-1.0,0.2) (0.8,-0.8) (2.4,-1.8)};
        \fill[blue!60] (-2.4,1.8) circle (4pt);
        \fill[red!60]  (-2.4,1.4) circle (4pt);
        \node[font=\footnotesize, blue!60] at (-2.4,2.2) {orbit 1};
        \node[font=\footnotesize, red!60] at (-2.4,1.0) {orbit 2};
        \node[font=\small, text width=4.5cm, align=center]
             at (0,-2.6) {How do nearby orbits\\diverge over time?};
    \end{scope}
    % Complexity box
    \begin{scope}[shift={(5.2,0)}]
        \draw[thick, fill=red!6, rounded corners=18pt]
             (-3.5,-3.8) rectangle (3.5,3.8);
        \node[font=\large\bfseries, red!70!black] at (0,3.0) {Complexity};
        \draw[thick, ->] (-2.5,-0.8) -- (2.5,-0.8) node[right, font=\tiny] {$n$};
        \draw[thick, ->] (-2.3,-1.1) -- (-2.3,2.2) node[above, font=\tiny] {time};
        \draw[thick, blue!60, domain=-2.0:2.2, samples=40]
             plot (\x, {0.12*(\x+2)*(\x+2) - 0.7});
        \draw[thick, red!60, domain=-2.0:1.8, samples=40]
             plot (\x, {0.28*(\x+2)*(\x+2) - 0.7});
        \node[font=\footnotesize, blue!60]  at (2.7, 0.8) {$\BigO(n)$};
        \node[font=\footnotesize, red!60]   at (2.3, 2.2) {$\BigO(n^2)$};
        \node[font=\small, text width=4.5cm, align=center]
             at (0,-2.6) {How do resource\\requirements compare?};
    \end{scope}
    % Connecting arrow
    \draw[ultra thick, <->, purple!70!black] (-2.0,0) -- (2.0,0);
    \node[font=\bfseries, purple!70!black, fill=white, inner sep=6pt,
          rounded corners=4pt, draw=purple!30]
         at (0,0) {Quasi-Metrics};
\end{tikzpicture}
\caption{Two worlds connected by quasi-metrics.  Dynamical systems study
orbit divergence; complexity theory studies resource growth.  The
asymmetry inherent in both settings is encoded by the quasi-metric
distance.}
\label{fig:two-worlds}
\end{figure}
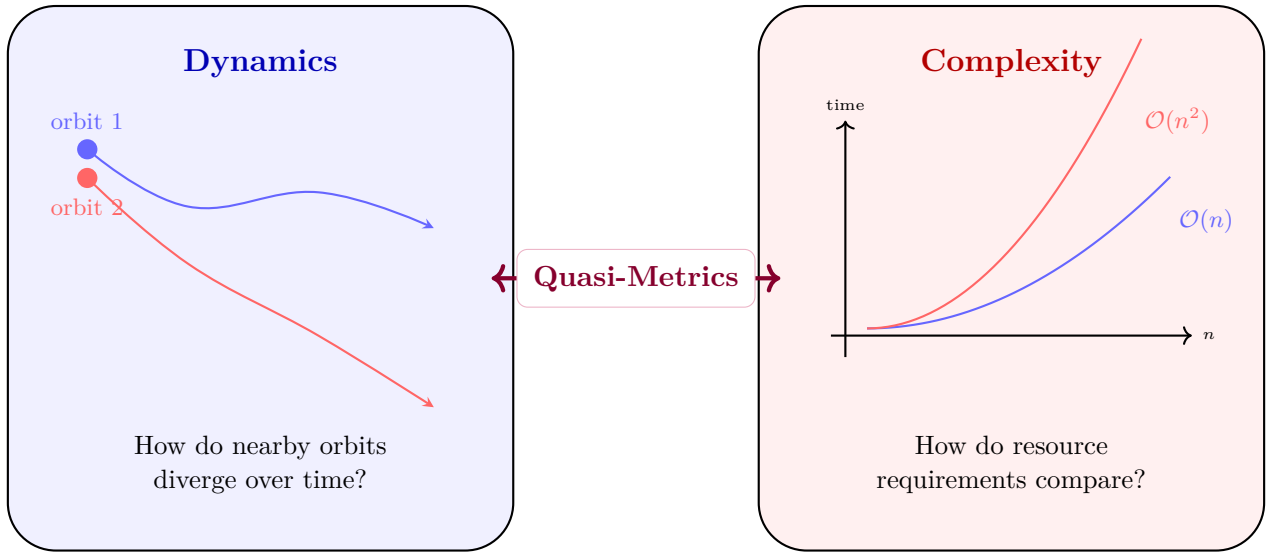

\subsection{Why quasi-metrics?}

Why quasi-metrics rather than ordinary metrics?  The reason is short.

In a standard metric space $(X,d)$, the symmetry axiom $d(x,y)=d(y,x)$
ensures that the cost of moving from~$x$ to~$y$ is the same as the
cost of moving from~$y$ to~$x$.  This is natural in many geometric
settings, but it is \emph{unnatural} in computational settings.  Consider
two algorithms with running times $f(n)=n$ and $g(n)=n^2$.  Given
the faster algorithm~$f$, one can trivially simulate the slower
algorithm~$g$ by wasting time.  But given the slower algorithm~$g$,
one cannot in general produce the faster algorithm~$f$ without effort.
The ``distance'' from fast to slow should therefore be zero (or small),
while the distance from slow to fast should be positive.

This is exactly what a quasi-metric provides.  By dropping the symmetry
axiom, quasi-metrics can encode directional costs, and the complexity
quasi-metric $d_\C$ does precisely this: $d_\C(f,g)=0$ whenever
$f(n)\le g(n)$ for all~$n$ (fast to slow is free), while
$d_\C(g,f)>0$ when $g$ is genuinely slower (slow to fast is costly).

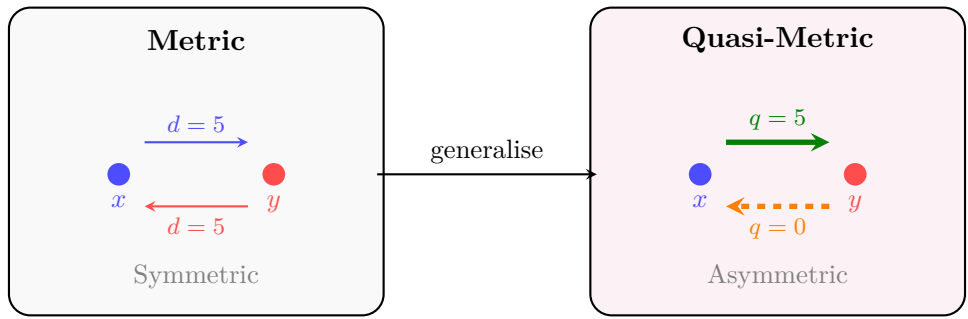
\begin{figure}[H]
\centering
\begin{tikzpicture}[scale=0.85, >=stealth]
  % Metric world
  \begin{scope}[shift={(-4.5,0)}]
    \draw[thick, rounded corners=8pt, fill=gray!5] (-2.9,-2.2) rectangle (2.9,2.6);
    \node[font=\bfseries] at (0,2.1) {Metric};
    \fill[blue!70] (-1.2,0) circle (5pt) node[below=4pt, font=\small] {$x$};
    \fill[red!70]  (1.2,0) circle (5pt) node[below=4pt, font=\small] {$y$};
    \draw[thick, blue!70, ->] (-0.8,0.5) -- (0.8,0.5)
      node[midway, above, font=\footnotesize] {$d=5$};
    \draw[thick, red!70, ->] (0.8,-0.5) -- (-0.8,-0.5)
      node[midway, below, font=\footnotesize] {$d=5$};
    \node[font=\small, gray] at (0,-1.6) {Symmetric};
  \end{scope}
  % Quasi-metric world
  \begin{scope}[shift={(4.5,0)}]
    \draw[thick, rounded corners=8pt, fill=purple!5] (-2.9,-2.2) rectangle (2.9,2.6);
    \node[font=\bfseries] at (0,2.1) {Quasi-Metric};
    \fill[blue!70] (-1.2,0) circle (5pt) node[below=4pt, font=\small] {$x$};
    \fill[red!70]  (1.2,0) circle (5pt) node[below=4pt, font=\small] {$y$};
    \draw[very thick, green!50!black, ->, line width=2.0pt] (-0.8,0.5) -- (0.8,0.5)
      node[midway, above, font=\footnotesize] {$q=5$};
    \draw[very thick, orange, ->, dashed, line width=2.0pt] (0.8,-0.5) -- (-0.8,-0.5)
      node[midway, below, font=\footnotesize] {$q=0$};
    \node[font=\small, gray] at (0,-1.6) {Asymmetric};
  \end{scope}
  \draw[thick, ->] (-1.7,0) -- (1.7,0) node[midway, above, font=\small] {generalise};
\end{tikzpicture}
\caption{From metrics to quasi-metrics.  Dropping symmetry allows the
distance to encode directional information.}
\label{fig:metric-vs-quasi}
\end{figure}

\subsection{Main results}

The paper establishes the following results about the dynamics of
the scaling transformation $\psi_\alpha(f)(n)=\alpha f(n)$ on the
complexity space $(\C,d_\C)$.

\begin{enumerate}[label=\textup{(\arabic*)},leftmargin=2.4em]
\item \emph{Expansiveness} (Theorem~\ref{thm:main-scaling}).
  The map $\psi_\alpha$ is expansive on $(\C,d_\C)$ if and only if
  $\alpha\ne 1$.

\item \emph{Stable and unstable sets}
  (Theorems~\ref{thm:stable-sets} and~\ref{thm:unstable-sets}).
  For $\alpha>1$, the $\delta$-stable set of $f$ under
  $\psi_\alpha$ coincides with the closed $d_\C$-ball
  $\{g:d_\C(f,g)\le\delta\}$; the $\delta$-unstable set admits the
  symmetric description in terms of the conjugate
  quasi-metric~$d_\C^t$.

\item \emph{Exact contraction rate}
  (Corollary~\ref{cor:exact-contraction}). For every $\alpha>0$ and
  every $n\ge 0$,
  \[
    d_\C(\psi_\alpha^n(f),\psi_\alpha^n(g))
    \;=\; \alpha^{-n}\,d_\C(f,g),
  \]
  with optimal multiplicative constant $C=1$.

\item \emph{Hierarchy and orbit separation}
  (Theorem~\ref{thm:hierarchy}). If $d_\C(g,f)>0$, the orbits of
  $f$ and $g$ under $\psi_\alpha$ ($\alpha\ne 1$) are eventually
  $d_\C^s$-separated.  The Hartmanis--Stearns
  condition $f(n)\log f(n)=o(g(n))$ is one sufficient source of
  such pairs.

\item \emph{A structural limit for topological entropy}
  (Proposition~\ref{prop:no-compact-invariant}). For $\alpha\ne 1$,
  every non-empty compact $\psi_\alpha$-invariant subset of
  $(\C,d_\C^s)$ has $d_\C^s$-diameter zero.  Consequently, the
  standard Bowen entropy of $\psi_\alpha$ on compact invariant
  sets is vacuous in this setting.
\end{enumerate}

\subsection{Organisation}

The paper is organised as follows.  Section~\ref{sec:quasi-metric}
recalls the basic theory of quasi-metric spaces, with examples and
motivation.  Section~\ref{sec:complexity-space} introduces the
complexity quasi-metric of Schellekens and establishes its
fundamental properties, including several illustrative computations.
Section~\ref{sec:expansive} defines expansive homeomorphisms in the
quasi-metric setting.  Section~\ref{sec:scaling} introduces the
scaling transformation and proves our main expansiveness result.
Section~\ref{sec:stable-unstable} develops the theory of stable and
unstable sets.  Section~\ref{sec:hyperbolicity} establishes
hyperbolicity.  Section~\ref{sec:hierarchy} connects orbit
separation to the time hierarchy theorem.
Section~\ref{sec:entropy} treats compact invariant sets and the
resulting obstruction to Bowen entropy.
Section~\ref{sec:conclusion} summarises the results.  All Python
implementations and SageMath verification scripts live in the
\href{\repourl}{companion repository} (see the
\href{\repourl/tree/main/code}{\texttt{code/}} directory); a list
of open problems and directions for future work is maintained
alongside the paper in
\href{\repourl/blob/main/OPEN-PROBLEMS.md}{\texttt{OPEN-PROBLEMS.md}}.

% ═══════════════════════════════════════════════════════════════
\section{Quasi-metric spaces}
\label{sec:quasi-metric}
% ═══════════════════════════════════════════════════════════════

We start by recalling the notion of a quasi-metric space.  The
theory has a long history, with major contributions by
K\"unzi~\cite{kunzi1995, kunzi2001},
Cobza\c{s}~\cite{cobzas2013}, and others.  Notation and conventions
follow Olela-Otafudu et al.~\cite{olela2024}.

\begin{definition}[Quasi-metric]\label{def:quasi-metric}
Let $X$ be a non-empty set.  A function
$q\colon X\times X\to[0,\infty)$ is a \emph{quasi-metric} on $X$
if it satisfies the following three axioms for all $x,y,z\in X$:
\begin{enumerate}[label=\textup{(Q\arabic*)},leftmargin=3em]
    \item $q(x,x)=0$; \label{Q1}
    \item $q(x,z)\le q(x,y)+q(y,z)$ \quad(triangle inequality); \label{Q2}
    \item $q(x,y)=0=q(y,x)\;\Rightarrow\;x=y$
      \quad($T_0$~separation). \label{Q3}
\end{enumerate}
The pair $(X,q)$ is called a \emph{quasi-metric space}.
\end{definition}

Notice that the only axiom ``missing'' compared with a metric is
symmetry: we do \emph{not} require $q(x,y)=q(y,x)$.  Dropping that
single axiom changes the theory more than one might expect.

Every quasi-metric $q$ gives rise to two natural companions.

\begin{definition}[Conjugate and symmetrization]\label{def:conjugate}
Let $(X,q)$ be a quasi-metric space.  The \emph{conjugate
quasi-metric} is $q^t(x,y):=q(y,x)$.  The \emph{symmetrization} is
$q^s(x,y):=\max\{q(x,y),q^t(x,y)\}=\max\{q(x,y),q(y,x)\}$.
\end{definition}

It is straightforward to verify that $q^t$ is again a quasi-metric
and that $q^s$ is a genuine metric on~$X$.  Thus every quasi-metric
space carries a canonical metric, obtained by taking the ``worst-case
direction'' of the asymmetric distance.

\begin{example}[Standard quasi-metric on $\R$]\label{ex:standard-qm}
Define $u\colon\R\times\R\to[0,\infty)$ by
\[
  u(x,y) \;=\; (y-x)^+ \;=\; \max\{0,\,y-x\}.
\]
Then $u$ is a quasi-metric:  \ref{Q1}~is immediate; \ref{Q2}~follows
from $(z-x)^+\le(y-x)^++(z-y)^+$; and \ref{Q3}~holds because
$u(x,y)=0=u(y,x)$ gives $y\le x$ and $x\le y$, hence $x=y$.

The conjugate is $u^t(x,y)=(x-y)^+$, and the symmetrization is
$u^s(x,y)=|x-y|$, the usual absolute-value metric.
\end{example}

The quasi-metric~$u$ has a vivid interpretation: \emph{going uphill
costs effort, while going downhill is free.}

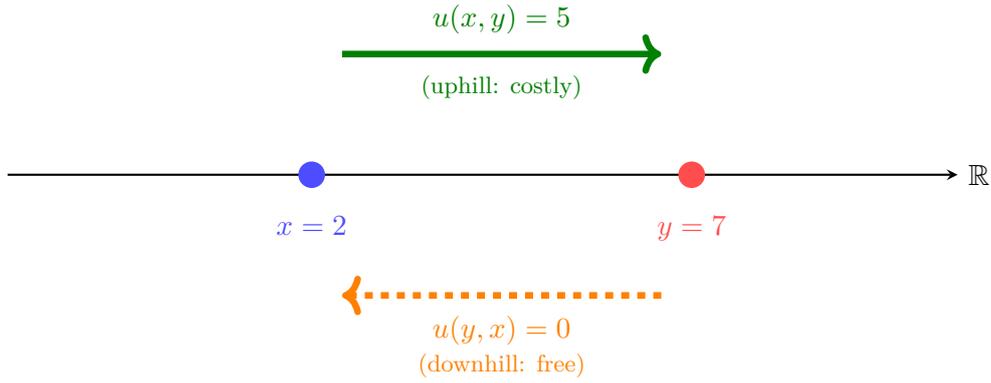
\begin{figure}[H]
\centering
\begin{tikzpicture}[scale=1.0]
    \draw[thick, ->, >=stealth] (-2.0,0) -- (10.5,0) node[right] {$\R$};
    % Tick marks
    \foreach \x in {2, 7} {
      \draw (\x,0.15) -- (\x,-0.15);
    }
    \fill[blue!70]  (2,0) circle (5pt);
    \node[below=12pt, blue!70, font=\bfseries] at (2,0) {$x=2$};
    \fill[red!70]   (7,0) circle (5pt);
    \node[below=12pt, red!70, font=\bfseries]  at (7,0) {$y=7$};
    % Forward arrow (uphill) -- moved higher
    \draw[->, very thick, green!50!black, line width=2.5pt]
         (2.4,1.6) -- (6.6,1.6);
    \node[above=4pt, green!50!black, font=\bfseries] at (4.5,1.6)
         {$u(x,y)=5$};
    \node[below=4pt, green!50!black, font=\footnotesize] at (4.5,1.6)
         {(uphill: costly)};
    % Backward arrow (downhill) -- moved lower
    \draw[->, very thick, orange, dashed, line width=2.5pt]
         (6.6,-1.6) -- (2.4,-1.6);
    \node[below=4pt, orange, font=\bfseries] at (4.5,-1.6)
         {$u(y,x)=0$};
    \node[below=18pt, orange, font=\footnotesize] at (4.5,-1.6)
         {(downhill: free)};
\end{tikzpicture}
\caption{The standard quasi-metric on $\R$: going from $x=2$ up to
$y=7$ costs $u(2,7)=5$, but going from $y=7$ down to $x=2$ is free:
$u(7,2)=0$.}
\label{fig:standard-qm}
\end{figure}

\begin{example}[Weighted quasi-metric on $\R$]\label{ex:weighted-qm}
For any weight $w>0$, define $u_w(x,y)=w\cdot(y-x)^+$.  This is
again a quasi-metric, and it models the situation where the cost of
going uphill is proportional to~$w$.  When $w=1$ we recover the
standard quasi-metric.  The symmetrization is $u_w^s(x,y)=w|x-y|$.
\end{example}

\begin{example}[Discrete quasi-metric]\label{ex:discrete-qm}
On any set $X$, define
\[
  q_d(x,y)=\begin{cases} 0 & \text{if } x=y,\\ 1 & \text{if } x\neq y.\end{cases}
\]
This is both a metric and a quasi-metric (the symmetric case).  It
illustrates that every metric is automatically a quasi-metric.
\end{example}

\begin{example}[Non-example: failing the triangle inequality]\label{ex:non-qm}
On $\R$, define $\rho(x,y)=(y-x)^2$ if $y\ge x$ and $\rho(x,y)=0$
if $y<x$.  Then $\rho$ satisfies~\ref{Q1} and~\ref{Q3}, but it
fails~\ref{Q2}: take $x=0$, $y=2$, $z=3$.  Then $\rho(x,z)=9$,
while $\rho(x,y)+\rho(y,z)=4+1=5<9$.  Hence $\rho$ is \emph{not}
a quasi-metric.  This illustrates that the triangle inequality is a
genuine restriction even in the asymmetric setting.
\end{example}

\begin{example}[Asymmetric topologies on $\{a,b,c\}$]\label{ex:asym-topologies}
Let $X=\{a,b,c\}$ with $q(a,b)=1$, $q(b,a)=3$, $q(a,c)=2$,
$q(c,a)=0$, $q(b,c)=1$, $q(c,b)=2$, and $q(x,x)=0$ for all~$x$.
One can verify that the triangle inequality holds.  The forward
topology $\tau_q$ has open ball $B_q(a,1.5)=\{a,b\}$, while the
conjugate topology $\tau_{q^t}$ has $B_{q^t}(a,1.5)=\{a,c\}$
(since $q^t(a,c)=q(c,a)=0<1.5$).  Thus the two topologies differ:
points that are ``close'' to~$a$ depend on the direction in which we
measure distance.  This is the hallmark of genuine asymmetry.
\end{example}

\begin{remark}[Topological considerations]\label{rem:topology}
A quasi-metric $q$ on $X$ generates a topology $\tau_q$ via the
open balls $B_q(x,\eps):=\{y\in X:q(x,y)<\eps\}$.  The conjugate
$q^t$ generates a potentially different topology $\tau_{q^t}$.
These two topologies coincide if and only if $q$ is symmetric, i.e.,
if $q$ is a metric.  The study of these asymmetric topologies is a
central theme in the work of K\"unzi~\cite{kunzi1995,kunzi2001} and
has deep connections to domain theory in computer
science~\cite{abramsky1994,scott1982}.
\end{remark}

% ═══════════════════════════════════════════════════════════════
\section{The complexity quasi-metric space}
\label{sec:complexity-space}
% ═══════════════════════════════════════════════════════════════

We now turn to the specific quasi-metric space this paper is
about.  The \emph{complexity space} was introduced by
Schellekens~\cite{schellekens1995} and further studied by Romaguera
and Schellekens~\cite{romaguera1999}; for the asymmetric-normed
sequence-space refinement see
Garc\'{i}a-Raffi--Romaguera--S\'anchez-P\'erez~\cite{garcia-raffi2002},
and for the extension to partial functions see
Romaguera--Schellekens--Valero~\cite{romaguera-schellekens-valero2011}.
It is the topological framework in which the resource-usage
functions associated with algorithms live.

\subsection{Definition and basic properties}

Throughout, $\N=\{1,2,3,\dots\}$.  Following
Schellekens~\cite{schellekens1995} and Romaguera and
Schellekens~\cite{romaguera1999}, let $\C$ denote the set of all
functions
\[
  f\colon\N\to(0,\infty)\quad\text{with}\quad
  \sum_{n=1}^\infty 2^{-n}/f(n)<\infty.
\]
Each such function represents the running-time profile of an
algorithm: $f(n)$ is the time required on inputs of size~$n$.  The
summability condition is the standing assumption of the framework:
it ensures that the weighting function $w_\C(f)=\sum_{n\ge 1}
2^{-n}/f(n)$ is finite, and hence that the quasi-metric defined
below takes only finite values.

\begin{remark}[Subspaces are also complexity spaces]\label{rem:subspace-convention}
Romaguera and Schellekens~\cite{romaguera1999} observe that any
subspace of $(\C,d_\C)$ is again called a complexity space.  Two
subspaces deserve special mention.  First, the
\emph{Turing-machine-runtime subspace}
$\C_{\ge 1}:=\{f\in\C: f(n)\ge 1 \text{ for all }n\}$ — for which
the bound $d_\C(f,g)\le 1$ of Theorem~\ref{thm:dc-props}(iii) below
holds uniformly — encodes the elementary fact that any computation
on input of size~$n$ takes at least one step.  Second, the original
formulation of~\cite{romaguera1999} permits the extended codomain
$(0,+\infty]$ with the convention $1/(+\infty)=0$, which makes
$\C$ a lattice under pointwise order; we work with the real-valued
subspace for notational lightness and because all of our examples
are real-valued.  Both restrictions are subspaces in the
Romaguera--Schellekens sense and inherit every result we prove.
\end{remark}

\begin{definition}[Complexity quasi-metric~\cite{schellekens1995}]
\label{def:dc}
The \emph{complexity quasi-metric} $d_\C\colon\C\times\C\to[0,\infty)$
is defined by
\[
  d_\C(f,g)
  \;=\; \sum_{n=1}^{\infty} 2^{-n}\,
    \max\!\left\{0,\;\frac{1}{g(n)}-\frac{1}{f(n)}\right\}.
\]
\end{definition}

The reciprocals $1/f(n)$ and $1/g(n)$ should be thought of as
\emph{efficiency measures}: a faster algorithm has a larger reciprocal.
The difference $1/g(n)-1/f(n)$ is positive when $g$ is more efficient
than~$f$ at input size~$n$, and the weighting $2^{-n}$ ensures
convergence of the series.

\begin{theorem}[Basic properties of $d_\C$]\label{thm:dc-props}
The following hold:
\begin{enumerate}[label=\textup{(\roman*)}]
  \item $d_\C$ is a quasi-metric on $\C$.
  \item $d_\C(f,g)=0$ if and only if $f(n)\le g(n)$ for all $n\in\N$.
  \item $d_\C(f,g)\le w_\C(g)=\sum_{n=1}^\infty 2^{-n}/g(n)$ for all
    $f,g\in\C$.  In particular, if $g(n)\ge 1$ for all $n$ (i.e.,
    $g\in\C_{\ge 1}$), then $d_\C(f,g)\le 1$.
\end{enumerate}
\end{theorem}

\begin{proof}
\textbf{(i)}  Axiom~\ref{Q1}: $d_\C(f,f)=\sum 2^{-n}\max\{0,0\}=0$.
The triangle inequality~\ref{Q2}: for each~$n$,
\[
  \max\left\{0,\frac{1}{h(n)}-\frac{1}{f(n)}\right\}
  \;\le\;
  \max\left\{0,\frac{1}{g(n)}-\frac{1}{f(n)}\right\}
  +\max\left\{0,\frac{1}{h(n)}-\frac{1}{g(n)}\right\},
\]
since the positive part is subadditive.  Multiplying by $2^{-n}$ and
summing gives $d_\C(f,h)\le d_\C(f,g)+d_\C(g,h)$.  Axiom~\ref{Q3}:
if $d_\C(f,g)=0=d_\C(g,f)$, then $f(n)\le g(n)$ and $g(n)\le f(n)$
for all~$n$, so $f=g$.

\textbf{(ii)}  $d_\C(f,g)=0$ iff each term is zero, iff
$1/g(n)\le 1/f(n)$ for all~$n$, iff $f(n)\le g(n)$.

\textbf{(iii)}  Each term satisfies
$2^{-n}\max\left\{0,\,1/g(n)-1/f(n)\right\}\le 2^{-n}/g(n)$, so
$d_\C(f,g)\le\sum_{n=1}^\infty 2^{-n}/g(n)=w_\C(g)$.  When
$g(n)\ge 1$ for all $n$, this bound is in turn dominated by
$\sum_{n=1}^\infty 2^{-n}=1$.
\end{proof}

Property~(ii) is the key asymmetry result: \emph{moving from a
faster function to a slower one is free}, because $f(n)\le g(n)$
(i.e., $f$ is faster) implies $d_\C(f,g)=0$.

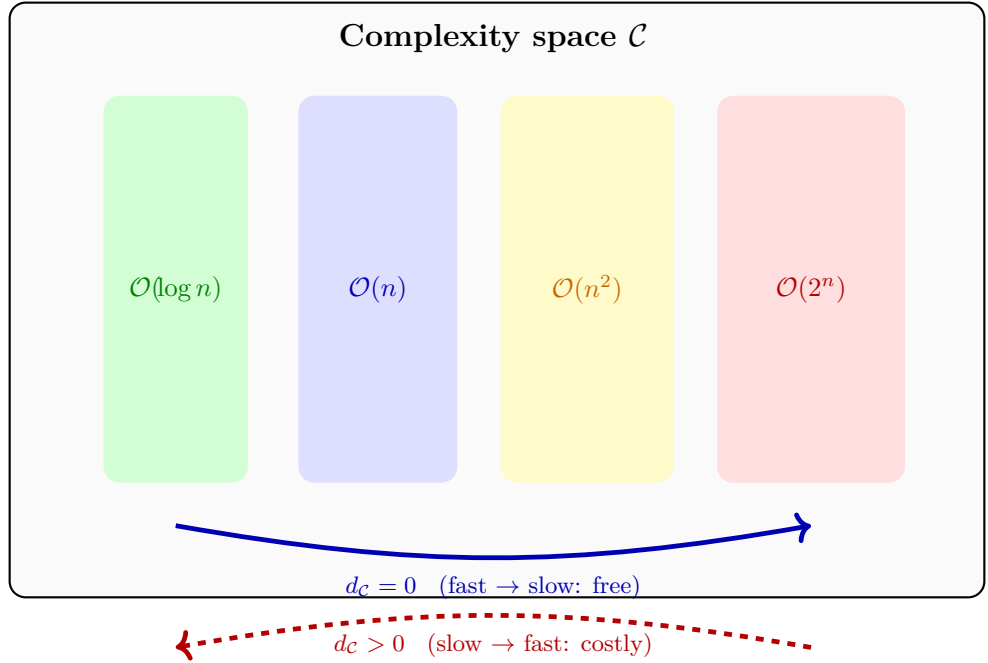
\begin{figure}[H]
\centering
\begin{tikzpicture}[scale=0.95]
    \draw[thick, fill=gray!4, rounded corners=6pt] (-1.0,-4.3) rectangle (12.5,4.0);
    \node[font=\large\bfseries] at (5.7,3.5) {Complexity space $\C$};
    % Classes
    \fill[green!20, opacity=.85, rounded corners=6pt]
         (0.3,-2.7) rectangle (2.3,2.7);
    \node[font=\small, green!50!black, align=center] at (1.3,0) {$\BigO(\!\log n)$};
    \fill[blue!15, opacity=.85, rounded corners=6pt]
         (3.0,-2.7) rectangle (5.2,2.7);
    \node[font=\small, blue!70!black] at (4.1,0) {$\BigO(n)$};
    \fill[yellow!30, opacity=.85, rounded corners=6pt]
         (5.8,-2.7) rectangle (8.2,2.7);
    \node[font=\small, orange!80!black] at (7.0,0) {$\BigO(n^2)$};
    \fill[red!15, opacity=.85, rounded corners=6pt]
         (8.8,-2.7) rectangle (11.4,2.7);
    \node[font=\small, red!70!black] at (10.1,0) {$\BigO(2^n)$};
    % Arrows
    \draw[->, very thick, blue!70!black, line width=1.8pt]
         (1.3,-3.3) to[bend right=10]
         node[below=2pt, font=\footnotesize, pos=0.5] {$d_\C=0$ \;\;(fast $\to$ slow: free)}
         (10.1,-3.3);
    \draw[->, very thick, red!70!black, dashed, line width=1.8pt]
         (10.1,-5.0) to[bend right=10]
         node[below=2pt, font=\footnotesize, pos=0.5] {$d_\C>0$ \;\;(slow $\to$ fast: costly)}
         (1.3,-5.0);
\end{tikzpicture}
\caption{The complexity landscape.  Moving from a faster class to a
slower one is free ($d_\C=0$); the reverse direction is costly ($d_\C>0$).}
\label{fig:complexity-landscape}
\end{figure}

\subsection{Illustrative examples}

A few worked examples make the definition concrete.

\begin{example}[Linear vs.\ quadratic]\label{ex:lin-quad}
Let $f(n)=n$ and $g(n)=n^2$.  Since $f(n)\le g(n)$ for all
$n\ge 1$, Theorem~\ref{thm:dc-props}(ii) gives $d_\C(f,g)=0$.
In the reverse direction,
\[
  d_\C(g,f)
  = \sum_{n=1}^{\infty}2^{-n}\max\!\left\{0,\frac{1}{n}-\frac{1}{n^2}\right\}
  = \sum_{n=2}^{\infty}2^{-n}\cdot\frac{n-1}{n^2},
\]
since the $n=1$ term vanishes ($\frac{1}{1}-\frac{1}{1}=0$).
We compute the first few partial sums to see how the series converges:
\begin{align*}
S_2 &= \tfrac{1}{4}\cdot\tfrac{1}{4} = 0.0625, \\
S_3 &= S_2 + \tfrac{1}{8}\cdot\tfrac{2}{9}
     = 0.0625 + 0.0278 = 0.0903, \\
S_4 &= S_3 + \tfrac{1}{16}\cdot\tfrac{3}{16}
     = 0.0903 + 0.0117 = 0.1020, \\
S_5 &= S_4 + \tfrac{1}{32}\cdot\tfrac{4}{25}
     = 0.1020 + 0.0050 = 0.1070.
\end{align*}
The series converges rapidly due to the $2^{-n}$ factor; by $n=10$
the partial sum is already $0.1108$, within $0.001$ of the limit.
Analytically, splitting $\frac{n-1}{n^2}=\frac{1}{n}-\frac{1}{n^2}$ and
using $\sum_{n=1}^\infty\frac{x^n}{n}=-\!\ln(1-x)$ and
$\sum_{n=1}^\infty\frac{x^n}{n^2}=\operatorname{Li}_2(x)$ at $x=\tfrac12$
gives the closed form
\[
  d_\C(g,f) = \ln 2 - \operatorname{Li}_2\!\bigl(\tfrac12\bigr)
  \approx 0.693 - 0.582 = 0.111.
\]
The exact value of the series can be confirmed symbolically
using SageMath; see
\href{\repourl/blob/main/code/sagemath/complexity_distances.sage}{\texttt{complexity\_distances.sage}}
in the companion repository.
The asymmetry is clear: moving from linear to quadratic is free, but
moving from quadratic to linear costs approximately~$0.111$.
\end{example}

\begin{example}[Logarithmic vs.\ linear]\label{ex:log-lin}
Let $f(n)=\ln(n+1)$ and $g(n)=n$.  Since $\ln(n+1)\le n$ for all
$n\ge 1$, we have $d_\C(f,g)=0$.  But $d_\C(g,f)>0$ because $g$
is slower.  The partial sums are:
\begin{align*}
S_1 &= \tfrac{1}{2}\bigl(\tfrac{1}{\ln 2}-1\bigr) \approx 0.2213, \quad
S_2 = S_1 + \tfrac{1}{4}\bigl(\tfrac{1}{\ln 3}-\tfrac{1}{2}\bigr) \approx 0.3179, \\
S_3 &\approx 0.3609, \quad
S_5 \approx 0.3991, \quad
S_{10} \approx 0.4165.
\end{align*}
The limit is $d_\C(g,f)\approx 0.417$; for exact symbolic evaluation see
\href{\repourl/blob/main/code/sagemath/complexity_distances.sage}{\texttt{complexity\_distances.sage}}.
\end{example}

\begin{example}[Equal functions]\label{ex:equal}
If $f=g$, then $d_\C(f,g)=d_\C(g,f)=0$.  The distance is
symmetric (and zero) for identical functions, as expected.
\end{example}

\begin{example}[Constant shift]\label{ex:constant-shift}
Let $f(n)=n$ and $g(n)=n+c$ for some constant $c>0$.  Then
$f(n)<g(n)$ for all~$n$, so $d_\C(f,g)=0$.  In the reverse
direction, $d_\C(g,f)=\sum_{n=1}^\infty 2^{-n}\cdot c/(n(n+c))$,
which is small but positive---reflecting the fact that~$g$ is only
slightly slower.
\end{example}

\begin{example}[Incomparable functions]\label{ex:incomparable}
Let $f(n)=n+(-1)^{n+1}$ and $g(n)=n$.  Then $f(1)=2>1=g(1)$ but
$f(2)=1<2=g(2)$, so neither $f(n)\le g(n)$ nor $g(n)\le f(n)$
for all~$n$.  (Note that $f$ alternates above and below~$g$: $f$
exceeds $g$ at odd indices and falls below at even indices.)
Consequently, \emph{both} $d_\C(f,g)>0$ and $d_\C(g,f)>0$.
Only odd-index terms contribute to $d_\C(f,g)$:
\[
  d_\C(f,g) = \sum_{\substack{n\ge 1\\n\text{ odd}}}
  \frac{2^{-n}}{n(n+1)}
  = \tfrac{1}{2}\cdot\tfrac{1}{2}
    +\tfrac{1}{8}\cdot\tfrac{1}{12}+\cdots
  \approx 0.262.
\]
Similarly, only even-index terms contribute to $d_\C(g,f)$:
\[
  d_\C(g,f) = \sum_{\substack{n\ge 2\\n\text{ even}}}
  \frac{2^{-n}}{n(n-1)}
  = \tfrac{1}{4}\cdot\tfrac{1}{2}
    +\tfrac{1}{16}\cdot\tfrac{1}{12}+\cdots
  \approx 0.131.
\]
The symmetrized distance is
$d_\C^s(f,g)=\max\{0.262,0.131\}\approx 0.262$.
Observe that $d_\C(f,g)\approx 2\,d_\C(g,f)$: the
``upward oscillation'' is costlier than the ``downward'' one,
reflecting the asymmetry of the quasi-metric.
The exact sums are verified symbolically in
\href{\repourl/blob/main/code/sagemath/incomparable_functions.sage}{\texttt{incomparable\_functions.sage}}.
\end{example}

\begin{example}[The bound is sharp]\label{ex:counter-bounded}
The bound in Theorem~\ref{thm:dc-props}(iii) is attained.  Take
$f(n)=1/n$ and $g(n)=1$ for all $n$.  Both lie in~$\C$: the
summability condition reads $\sum 2^{-n}/f(n)=\sum n\cdot 2^{-n}=2$
for $f$, and $w_\C(g)=\sum 2^{-n}=1$ for $g$.  Since
$f(n)\le g(n)$ for all $n\ge 1$, Theorem~\ref{thm:dc-props}(ii)
gives $d_\C(f,g)=0$: $f$ is faster than $g$ pointwise, so moving
from $f$ to $g$ is ``free.''  In the reverse direction the $n=1$ term vanishes
($f(1)=1=g(1)$) and
\[
  d_\C(g,f)
  \;=\; \sum_{n=2}^\infty 2^{-n}\,\Bigl(\tfrac{1}{f(n)}-\tfrac{1}{g(n)}\Bigr)
  \;=\; \sum_{n=2}^\infty 2^{-n}(n-1)
  \;=\; \tfrac{1}{2}\sum_{m=1}^\infty m\cdot 2^{-m}
  \;=\; 1.
\]
This saturates the universal bound $d_\C(g,f)\le w_\C(g)$ of
Theorem~\ref{thm:dc-props}(iii); here $g\in\C_{\ge 1}$, so
$w_\C(g)=1$.  The example witnesses the asymmetry: the cost of
swapping a faster function for a slower one is bounded by the
weight of the target, with equality possible.
\end{example}

\subsection{The conjugate and symmetrized complexity distances}

Since $d_\C$ is a quasi-metric, we automatically obtain the conjugate
and symmetrization.

\begin{proposition}\label{prop:conjugate-dc}
The conjugate of the complexity quasi-metric is
\[
  d_\C^t(f,g) \;=\; d_\C(g,f)
  \;=\; \sum_{n=1}^{\infty}2^{-n}\max\!\left\{0,\frac{1}{f(n)}-\frac{1}{g(n)}\right\}.
\]
The symmetrization is $d_\C^s(f,g)=\max\{d_\C(f,g),d_\C(g,f)\}$.
\end{proposition}

\begin{proof}
By Definition~\ref{def:conjugate}, $d_\C^t(f,g)=d_\C(g,f)$.
Substituting $g$ for the first argument and $f$ for the second
in Definition~\ref{def:dc} yields the stated formula.  That
$d_\C^s=\max\{d_\C,d_\C^t\}$ is a metric follows from the
general theory (Definition~\ref{def:conjugate}).
\end{proof}

The symmetrized distance $d_\C^s$ is a genuine metric on $\C$ and
measures the ``worst-case directional cost'' between two complexity
functions.

\begin{example}[Symmetrization of linear vs.\ quadratic]\label{ex:sym}
From Example~\ref{ex:lin-quad}, $d_\C(f,g)=0$ and
$d_\C(g,f)\approx 0.111$, so $d_\C^s(f,g)\approx 0.111$.
\end{example}

\subsection{Computing the complexity quasi-metric}

The following algorithm approximates $d_\C(f,g)$ numerically.
Since the series has infinitely many terms, we truncate at
$N$~terms; the exponential decay of $2^{-n}$ guarantees rapid
convergence.

\begin{algorithm}[H]
\DontPrintSemicolon
\SetAlgoLined
\KwIn{Functions $f,g$; truncation parameter $N$}
\KwOut{Approximation of $d_\C(f,g)$}
$S \leftarrow 0$\;
\For{$n \leftarrow 1$ \KwTo $N$}{
    $\Delta \leftarrow 1/g(n) - 1/f(n)$\;
    \If{$\Delta > 0$}{
        $S \leftarrow S + 2^{-n}\,\Delta$\;
    }
}
\Return{$S$}\;
\caption{Compute $d_\C(f,g)$}
\label{alg:dc}
\end{algorithm}

\begin{remark}[Convergence rate]
The truncation error after $N$ terms is at most $2^{-N}$, since each
omitted term contributes at most $2^{-n}$.  In practice, $N=80$ gives
accuracy well beyond double-precision floating-point.
\end{remark}

A Python implementation of Algorithm~\ref{alg:dc} is provided in
\href{\repourl/blob/main/code/python/complexity_distance.py}{\texttt{complexity\_distance.py}}.
For many common function pairs, the infinite series admits a
closed-form evaluation via symbolic algebra; we use SageMath for
exact verification throughout (see
\href{\repourl/tree/main/code/sagemath}{\texttt{code/sagemath/}}).

% ═══════════════════════════════════════════════════════════════
\section{Expansive homeomorphisms}
\label{sec:expansive}
% ═══════════════════════════════════════════════════════════════

The dynamical side of the theory comes next.  An expansive
homeomorphism is, roughly, a map that ``spreads things out'': no
two distinct points stay close forever under iteration.

\subsection{Definition and motivation}

In a standard metric space, expansiveness was introduced by
Utz~\cite{utz1950}: a homeomorphism $\psi\colon X\to X$ is
\emph{expansive} if there exists a constant $\delta>0$ such that for
every pair of distinct points $x\neq y$, some iterate $\psi^n$
separates them by more than~$\delta$.  The constant~$\delta$ is
called the \emph{expansive constant}.

Olela-Otafudu et al.~\cite{olela2024} extended this to
quasi-metric spaces, where the asymmetry of the distance introduces
new subtleties.

\begin{definition}[$q$-expansive homeomorphism~{\cite{olela2024}}]
\label{def:expansive}
Let $(X,q)$ be a $T_0$ quasi-metric space and let
$\psi\colon X\to X$ be a homeomorphism with respect to the
metric topology $\tau_{q^s}$ generated by the symmetrization
$q^s=\max\{q,q^t\}$.  We say that $\psi$ is \emph{$q$-expansive} if
there exists a constant $\delta>0$ (an \emph{expansive constant})
such that for every pair of distinct points $x\ne y$ in $X$, there
exists $n\in\Z$ with
\[
  \max\{\,q(\psi^n(x),\psi^n(y)),\;q(\psi^n(y),\psi^n(x))\,\}
  \;>\;\delta.
\]
\end{definition}

Note that we allow both positive and negative iterates: the
separation may occur in the future or in the past.  This is essential
in the quasi-metric setting, because the asymmetry of $q$ means that
forward and backward iterates may behave very differently.

\begin{figure}[H]
\centering
\begin{tikzpicture}[scale=0.9, >=stealth]
  % Box
  \draw[thick, fill=blue!3, rounded corners=12pt] (-1.5,-3.0) rectangle (12.5,4.0);
  \node[font=\bfseries] at (5.5,3.5) {Expansive homeomorphism};
  % Two points
  \fill[blue!70] (1.2,1.4) circle (5pt) node[above=4pt, font=\small] {$x$};
  \fill[red!70]  (1.8,1.0) circle (5pt) node[below=4pt, font=\small] {$y$};
  % Arrow
  \draw[thick, ->] (2.6,1.2) -- (4.0,1.2) node[midway, above, font=\footnotesize] {$\psi^n$};
  % Separated points
  \fill[blue!70] (5.2,2.6) circle (5pt) node[above=4pt, font=\small] {$\psi^n(x)$};
  \fill[red!70]  (5.2,-1.6) circle (5pt) node[below=4pt, font=\small] {$\psi^n(y)$};
  % Distance
  \draw[<->, thick, purple, line width=1.5pt] (5.7,2.4) -- (5.7,-1.4);
  \node[right, purple, font=\small] at (5.8,0.5) {$q > \delta$};
  % Note
  \draw[thick, dashed, gray] (7.4,1.2) -- (11.0,1.2);
  \node[font=\small, text width=3.8cm, align=center] at (9.2,1.2)
    {Every pair eventually separated by $>\delta$};
\end{tikzpicture}
\caption{An expansive homeomorphism: the orbits of any two distinct
points are eventually separated by more than~$\delta$.}
\label{fig:expansive-idea}
\end{figure}
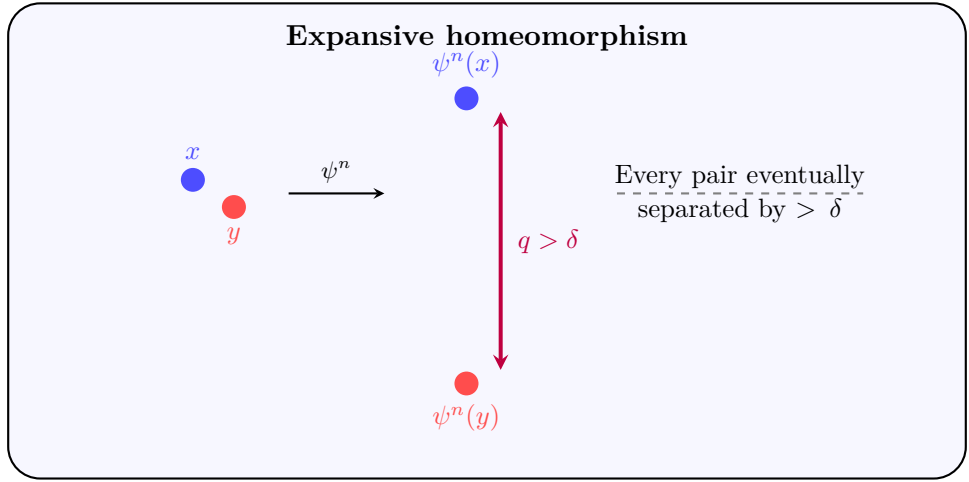

\subsection{Equivalences in the quasi-metric setting}

A key result of~\cite{olela2024} is that expansiveness with respect
to~$q$ is equivalent to expansiveness with respect to the conjugate~$q^t$.

\begin{theorem}[\cite{olela2024}]\label{thm:olela-equiv}
Let $(X,q)$ be a quasi-metric space and $\psi\colon X\to X$ a
homeomorphism.  Then:
\begin{enumerate}[label=\textup{(\roman*)}]
  \item $\psi$ is $q$-expansive if and only if $\psi$ is $q^t$-expansive.
  \item If $\psi$ is $q$-expansive, then $\psi$ is $q^s$-expansive.
    The converse is false in general.
\end{enumerate}
\end{theorem}

Part~(i) is striking: the direction of asymmetry does not matter
for whether expansive behaviour \emph{exists}, though it may
change the expansive constant.  Part~(ii) says that quasi-metric
expansiveness is strictly stronger than metric expansiveness.

\begin{example}[Non-equivalence of (ii)]
\label{ex:converse-false}
Consider $X=\{0,1\}^\Z$ (the full two-shift) with the quasi-metric
$q(x,y)=\sum_{n\ge 0}2^{-n}|x_n-y_n|$ (only non-negative indices).
The shift map is $q^s$-expansive but not $q$-expansive, because
forward-only distances cannot detect differences in the past.
\end{example}

\begin{example}[Illustrating part~(i): $q$-expansive $\Leftrightarrow$ $q^t$-expansive]
\label{ex:q-qt-equiv}
Consider $f(n)=n$ and $g(n)=2n$ on $(\C,d_\C)$ with $\psi_2$.
We have $d_\C(f,g)=0$ (since $n\le 2n$) and
$d_\C(g,f)=\sum_{n=1}^\infty 2^{-n}\cdot\frac{1}{2n}
=\frac{1}{2}\sum_{n=1}^\infty\frac{(1/2)^n}{n}=\frac{1}{2}\ln 2
\approx 0.347$
(partial sums: $S_1=0.250$, $S_2=0.313$, $S_3=0.333$,
$S_5=0.344$; exact value: $\frac{1}{2}\ln 2$).
The backward iterates give $d_\C(\psi_2^{-k}(g),\psi_2^{-k}(f))
=2^k\cdot 0.347$, which exceeds any $\delta$ for $k$ large enough.
Thus $\psi_2$ is $d_\C$-expansive.  For the conjugate:
$d_\C^t(f,g)=d_\C(g,f)\approx 0.347>0$, and
$d_\C^t(\psi_2^{-k}(f),\psi_2^{-k}(g))=2^k\cdot 0.347\to\infty$.
So $\psi_2$ is also $d_\C^t$-expansive, confirming part~(i) of
Theorem~\ref{thm:olela-equiv}.
\end{example}

\begin{example}[The identity is not expansive]\label{ex:non-expansive}
The identity map $\mathrm{id}\colon\C\to\C$ is trivially a
$\tau_{d_\C^s}$-homeomorphism.  For any $\delta>0$, choose
$f,g\in\C$ with $0<d_\C^s(f,g)<\delta$ (for example, two distinct
constant functions $f\equiv k_1$, $g\equiv k_2$ with $|k_1-k_2|$
small).  Every iterate of $\mathrm{id}$ leaves the distance
unchanged, so $d_\C^s(\mathrm{id}^n(f),\mathrm{id}^n(g))<\delta$
for all $n\in\Z$.  Hence $\mathrm{id}$ is not expansive.  This
example shows that the conclusion of
Theorem~\ref{thm:main-scaling} can fail even for
$\tau_{d_\C^s}$-homeomorphisms when $\alpha=1$.
\end{example}

\subsection{Checking expansiveness numerically}

Given two functions $f,g$ and a candidate map $\psi$, we can
numerically check whether their orbits separate.  The idea is simple:
iterate $\psi$ both forward and backward and check whether the
distance exceeds~$\delta$ at some iterate.

\begin{algorithm}[H]
\DontPrintSemicolon
\SetAlgoLined
\KwIn{Functions $f\neq g$; map $\psi$; candidate $\delta$;
      iteration bound $M$}
\KwOut{\texttt{True} if separation $>\delta$ found; the iterate $n$}
\For{$n \leftarrow -M$ \KwTo $M$}{
    $d \leftarrow q(\psi^n(f),\,\psi^n(g))$\;
    \If{$d > \delta$}{
        \Return{\texttt{True}, $n$}\;
    }
}
\Return{\texttt{False}, \texttt{None}}\;
\caption{Check expansive separation}
\label{alg:check-exp}
\end{algorithm}

A Python implementation is given in
\href{\repourl/blob/main/code/python/expansiveness_check.py}{\texttt{expansiveness\_check.py}}.

% ═══════════════════════════════════════════════════════════════
\section{The scaling transformation}
\label{sec:scaling}
% ═══════════════════════════════════════════════════════════════

The main dynamical object of the paper is the \emph{scaling
transformation} on the complexity space---the simplest non-trivial
map that respects the multiplicative structure of running-time
functions.

\subsection{Definition and basic properties}

\begin{definition}[Scaling transformation]\label{def:scaling}
For $\alpha>0$, the \emph{scaling transformation}
$\psi_\alpha\colon\C\to\C$ is defined by
\[
  \psi_\alpha(f)(n) \;=\; \alpha\cdot f(n).
\]
Its $k$-th iterate is $\psi_\alpha^k(f)(n)=\alpha^k f(n)$.
\end{definition}

The map $\psi_\alpha$ multiplies every running-time value by the
constant~$\alpha$.  When $\alpha>1$, this makes algorithms
``slower'' (larger running times); when $0<\alpha<1$, it makes them
``faster.''  When $\alpha=1$, it is the identity.

\begin{example}[Scaling a linear function]
If $f(n)=n$ and $\alpha=2$, then $\psi_2(f)(n)=2n$,
$\psi_2^2(f)(n)=4n$, $\psi_2^3(f)(n)=8n$, and in general
$\psi_2^k(f)(n)=2^k n$.  The orbit of~$f$ under $\psi_2$
consists of all functions of the form $2^k n$ for $k\in\Z$.
\end{example}

\begin{example}[Scaling a quadratic function]
If $g(n)=n^2$ and $\alpha=3$, then $\psi_3^k(g)(n)=3^k n^2$.
The orbit consists of functions $3^k n^2$, which are all
quadratic but with different leading coefficients.
\end{example}

The key algebraic property of $\psi_\alpha$ with respect to $d_\C$
is that it acts as a \emph{Lipschitz contraction} (when $\alpha>1$)
or \emph{expansion} (when $\alpha<1$).

\begin{lemma}\label{lem:scaling-lip}
For any $\alpha>0$ and any $f,g\in\C$,
\[
  d_\C(\psi_\alpha(f),\psi_\alpha(g))
  \;=\; \frac{1}{\alpha}\,d_\C(f,g).
\]
\end{lemma}

\begin{proof}
We compute directly:
\begin{align*}
  d_\C(\psi_\alpha(f),\psi_\alpha(g))
  &= \sum_{n=1}^{\infty}2^{-n}\max\!\left\{0,
     \frac{1}{\alpha g(n)}-\frac{1}{\alpha f(n)}\right\}\\
  &= \sum_{n=1}^{\infty}2^{-n}\cdot\frac{1}{\alpha}\,
     \max\!\left\{0,\frac{1}{g(n)}-\frac{1}{f(n)}\right\}
  \;=\;\frac{1}{\alpha}\,d_\C(f,g). \qedhere
\end{align*}
\end{proof}

By induction, $d_\C(\psi_\alpha^k(f),\psi_\alpha^k(g))
=\alpha^{-k}\,d_\C(f,g)$ for all $k\ge 0$.

\begin{example}[Numerical verification of Lemma~\ref{lem:scaling-lip}]
\label{ex:lip-numerical}
Let $f(n)=n$, $g(n)=n^2$, and $\alpha=3$.  We have
$d_\C(g,f)\approx 0.111$.  After scaling:
$\psi_3(f)(n)=3n$ and $\psi_3(g)(n)=3n^2$.  Then
$d_\C(\psi_3(g),\psi_3(f))=d_\C(3n^2,3n)
=\sum_{n=1}^\infty 2^{-n}\max\!\left\{0,\frac{1}{3n}-\frac{1}{3n^2}\right\}
=\frac{1}{3}\cdot d_\C(g,f)\approx 0.037$, which matches
$\frac{1}{\alpha}\cdot 0.111=0.037$ exactly
as confirmed by SageMath
(\href{\repourl/blob/main/code/sagemath/partial_sums.sage}{\texttt{partial\_sums.sage}}).
\end{example}

\begin{remark}[Group structure]
The scaling maps form a multiplicative group:
$\psi_\alpha\circ\psi_\beta=\psi_{\alpha\beta}$ and
$\psi_\alpha^{-1}=\psi_{1/\alpha}$.  In particular, the family
$\{\psi_\alpha:\alpha>0\}$ is isomorphic to $(\R_{>0},\cdot)$.
The Lipschitz property of Lemma~\ref{lem:scaling-lip} extends
to this group action:
$d_\C(\psi_\alpha\circ\psi_\beta(f),\psi_\alpha\circ\psi_\beta(g))
=\frac{1}{\alpha\beta}\,d_\C(f,g)$.
\end{remark}

\begin{remark}[Interpretation]
When $\alpha>1$, forward iterates of $\psi_\alpha$ bring functions
\emph{closer} in the $d_\C$ distance (contraction by factor
$1/\alpha$ per step).  Backward iterates push them \emph{apart}
(expansion by factor $\alpha$ per step).  When $0<\alpha<1$, the
roles reverse.
\end{remark}

\subsection{Main theorem: Expansiveness of scaling}

The main result says that the scaling map is expansive exactly
when it is non-trivial.

\begin{lemma}[$\psi_\alpha$ is a homeomorphism]\label{lem:psi-homeo}
For every $\alpha>0$, the map $\psi_\alpha\colon\C\to\C$ is a
bijection with inverse $\psi_{1/\alpha}$, and both $\psi_\alpha$
and $\psi_{1/\alpha}$ are Lipschitz with respect to $d_\C^s$.  In
particular, $\psi_\alpha$ is a $\tau_{d_\C^s}$-homeomorphism.
\end{lemma}

\begin{proof}
First, $\psi_\alpha$ maps $\C$ to itself: for $f\in\C$,
$\sum_{n\ge 1}2^{-n}/(\alpha f(n))=(1/\alpha)\sum_{n\ge 1}2^{-n}/f(n)
<\infty$, so $\psi_\alpha(f)\in\C$.  The inverse identity
$\psi_\alpha\circ\psi_{1/\alpha}=\mathrm{id}$ is immediate.  By
Lemma~\ref{lem:scaling-lip} applied to both $d_\C$ and $d_\C^t$,
$d_\C^s(\psi_\alpha(f),\psi_\alpha(g))=(1/\alpha)d_\C^s(f,g)$,
giving the Lipschitz bound with constant $1/\alpha$; the same
applies to $\psi_{1/\alpha}$ with constant $\alpha$.  Both maps are
therefore $\tau_{d_\C^s}$-continuous, and $\psi_\alpha$ is a
homeomorphism.
\end{proof}

\begin{theorem}[Main theorem]\label{thm:main-scaling}
The scaling transformation $\psi_\alpha$ is expansive on $(\C,d_\C)$
if and only if $\alpha\neq 1$.
\end{theorem}

\begin{proof}
We consider three cases.

\textbf{Case 1: $\alpha=1$.}  The map $\psi_1$ is the identity, so
$d_\C(\psi_1^n(f),\psi_1^n(g))=d_\C(f,g)$ for all~$n$.  If
$d_\C(f,g)>0$, the distance is constant and may be smaller than any
proposed~$\delta$; if $d_\C(f,g)=0$ but $d_\C(g,f)>0$ (i.e., $f$ is
faster than~$g$), then $d_\C(\psi_1^n(f),\psi_1^n(g))=0$ for all~$n$
and no separation occurs.  Hence $\psi_1$ is not expansive.

\textbf{Case 2: $\alpha>1$.}  Take any $\delta\in(0,1)$ and let
$f\ne g$ in $\C$.  Since $(\C,d_\C)$ is $T_0$, at least one of
$d_\C(f,g)$ or $d_\C(g,f)$ is positive.  Note that
$\psi_\alpha^{-1}=\psi_{1/\alpha}$, so iterating
Lemma~\ref{lem:scaling-lip} with $\alpha\mapsto 1/\alpha$ gives,
for all $h_1,h_2\in\C$ and all $k\ge 0$,
\begin{equation}\label{eq:backward-iterate-bound}
  d_\C\!\left(\psi_\alpha^{-k}(h_1),\,\psi_\alpha^{-k}(h_2)\right)
  \;=\; \alpha^{k}\,d_\C(h_1,h_2).
\end{equation}

If $d_\C(f,g)>0$, then~\eqref{eq:backward-iterate-bound} applied to
$(h_1,h_2)=(f,g)$ gives some $k\ge 0$ with
$d_\C(\psi_\alpha^{-k}(f),\psi_\alpha^{-k}(g))=\alpha^k d_\C(f,g)>\delta$,
and Definition~\ref{def:expansive} is satisfied with $n=-k$.

If instead $d_\C(f,g)=0$ but $d_\C(g,f)>0$, apply
\eqref{eq:backward-iterate-bound} to $(h_1,h_2)=(g,f)$: for some
$k\ge 0$,
$d_\C(\psi_\alpha^{-k}(g),\psi_\alpha^{-k}(f))=\alpha^k d_\C(g,f)>\delta$.
Since $d_\C(\psi_\alpha^{-k}(g),\psi_\alpha^{-k}(f))
=d_\C^t(\psi_\alpha^{-k}(f),\psi_\alpha^{-k}(g))$, the symmetric
condition in Definition~\ref{def:expansive} (with $q=d_\C$, $n=-k$)
is satisfied.

\textbf{Case 3: $0<\alpha<1$.}  Take any $\delta\in(0,1)$.  Let
$f\ne g$ in $\C$.  At least one of $d_\C(f,g)$ or $d_\C(g,f)$ is
positive.

If $d_\C(f,g)>0$, then for forward iterates:
$d_\C(\psi_\alpha^k(f),\psi_\alpha^k(g))=\alpha^{-k}d_\C(f,g)>\delta$
for $k$ large enough.  Definition~\ref{def:expansive} is satisfied
with $n=k$.

If instead $d_\C(f,g)=0$ but $d_\C(g,f)>0$, apply the same argument
to the pair $(g,f)$: $d_\C(\psi_\alpha^k(g),\psi_\alpha^k(f))
=\alpha^{-k}d_\C(g,f)>\delta$ for $k$ large, and the symmetric
condition in Definition~\ref{def:expansive} is satisfied with
$n=k$.
\end{proof}

\begin{example}[Expansiveness for $\alpha=2$]\label{ex:expansive-alpha2}
Consider $f(n)=n$ and $g(n)=n+1$. For $\alpha=2$, we have $d_\C(f,g)=0$
(since $f(n)<g(n)$ for all $n$) but
\[
  d_\C(g,f)=\sum_{n=1}^\infty
  \frac{2^{-n}}{n(n+1)}.
\]
Computing partial sums: $S_1=\frac{1}{4}=0.250$,
$S_2=S_1+\frac{1}{24}\approx 0.292$,
$S_3\approx 0.302$, $S_5\approx 0.306$, converging to
$\approx 0.307$
(see \href{\repourl/blob/main/code/sagemath/complexity_distances.sage}{\texttt{complexity\_distances.sage}}).
The backward iterates give:
\[
d_\C(\psi_2^{-k}(g),\psi_2^{-k}(f)) = 2^k \cdot 0.307
\]
For $\delta=0.5$, we need $2^k \cdot 0.307 > 0.5$, i.e., $k \geq 1$.
Indeed, $d_\C(\psi_2^{-1}(g),\psi_2^{-1}(f)) \approx 0.614 > 0.5$.
Thus $\psi_2$ is expansive with expansive constant $\delta=0.5$.
\end{example}

\begin{example}[Non-expansiveness for $\alpha=1$]\label{ex:non-expansive-alpha1}
Take $f(n)=n$ and $g(n)=n^2$. For $\alpha=1$, $\psi_1$ is the identity, 
so $d_\C(\psi_1^n(f),\psi_1^n(g)) = d_\C(f,g)=0$ for all $n$, 
while $d_\C(\psi_1^n(g),\psi_1^n(f)) = d_\C(g,f)\approx 0.111$ for all $n$.
No matter how large $n$ is, $d_\C(f,g)$ remains $0$, so no separation occurs 
in that direction. Therefore $\psi_1$ is not expansive.
\end{example}

\begin{figure}[H]
\centering
\begin{tikzpicture}[scale=0.8]
    % Title with more space above
    \node[font=\small\bfseries] at (5.5,6.2) {Three regimes of the scaling transformation $\psi_\alpha$};
    
    % Case alpha > 1 - with shorter axes
    \begin{scope}[shift={(-5.5,0)}]
        \draw[thick, ->, >=stealth] (-0.8,0) -- (3.8,0) node[right, font=\footnotesize] {iterate $k$};
        \draw[thick, ->, >=stealth] (0,-0.8) -- (0,3.8);
        \node[font=\small\bfseries, blue!80!black] at (1.5,-1.0) {(a) $\alpha > 1$};
        
        % Exponential growth curve (backward iterates) - shorter domain
        \draw[very thick, blue!80!black, domain=0:3.2, samples=30]
             plot (\x, {0.35*exp(0.65*\x)});
        
        % Delta threshold line - positioned appropriately
        \draw[thick, orange!80!black, dashed] (-0.5,2.0) -- (3.8,2.0);
        \node[orange!80!black, right, font=\footnotesize] at (3.8,2.0) {$\delta$};
        
        % Annotations - moved closer to curve
        \node[font=\tiny, align=center, text width=2.5cm, blue!70!black] at (1.5,2.8) 
             {$d_\C(\psi_\alpha^{-k}(f),\psi_\alpha^{-k}(g))$\\$=\alpha^k \cdot d_\C(f,g)$};
        
        \node[font=\footnotesize, gray, align=center] at (1.5,-1.8) 
             {Backward orbits diverge};
             
        % Example points - fewer points
        \foreach \k in {0,1,2,3} {
            \fill[blue!80!black] (\k, {0.35*exp(0.65*\k)}) circle (2.5pt);
        }
        
        % Iterate labels
        \foreach \k in {0,1,2,3} {
            \node[font=\tiny, below] at (\k,-0.15) {$\k$};
        }
        
        \node[font=\tiny, above left] at (0,0.35) {$d_\C(f,g)$};
    \end{scope}
    
    % Vertical space between cases
    \draw[white] (4.2,0) -- (4.8,0); % Invisible spacer
    
    % Case alpha = 1 - with shorter axes
    \begin{scope}[shift={(0,0)}]
        \draw[thick, ->, >=stealth] (-0.8,0) -- (3.8,0) node[right, font=\footnotesize] {iterate $k$};
        \draw[thick, ->, >=stealth] (0,-0.8) -- (0,3.8);
        \node[font=\small\bfseries, purple!80!black] at (1.5,-1.0) {(b) $\alpha = 1$};
        
        % Constant line (identity map) - shorter
        \draw[very thick, purple!80!black, line width=2.2pt] (0,1.5) -- (3.5,1.5);
        
        % Delta threshold line
        \draw[thick, orange!80!black, dashed] (-0.5,2.0) -- (3.8,2.0);
        \node[orange!80!black, right, font=\footnotesize] at (3.8,2.0) {$\delta$};
        
        % Annotations
        \node[font=\tiny, align=center, text width=2.5cm, purple!70!black] at (1.5,2.8) 
             {$d_\C(\psi_1^k(f),\psi_1^k(g))$\\$= d_\C(f,g)$ constant};
        
        \node[font=\footnotesize, gray, align=center] at (1.5,-1.8) 
             {No separation};
             
        % Iterate labels
        \foreach \k in {0,1,2,3} {
            \node[font=\tiny, below] at (\k,-0.15) {$\k$};
            \fill[purple!80!black] (\k,1.5) circle (2.5pt);
        }
        
        \node[font=\tiny, left] at (0,1.5) {$d_\C(f,g)$};
    \end{scope}
    
    % Vertical space between cases
    \draw[white] (4.2,0) -- (4.8,0); % Invisible spacer
    
    % Case alpha < 1 - with shorter axes
    \begin{scope}[shift={(5.5,0)}]
        \draw[thick, ->, >=stealth] (-0.8,0) -- (3.8,0) node[right, font=\footnotesize] {iterate $k$};
        \draw[thick, ->, >=stealth] (0,-0.8) -- (0,3.8);
        \node[font=\small\bfseries, green!60!black] at (1.5,-1.0) {(c) $\alpha < 1$};
        
        % Exponential growth curve (forward iterates) - shorter domain
        \draw[very thick, green!60!black, domain=0:3.2, samples=30]
             plot (\x, {0.35*exp(0.65*\x)});
        
        % Delta threshold line
        \draw[thick, orange!80!black, dashed] (-0.5,2.0) -- (3.8,2.0);
        \node[orange!80!black, right, font=\footnotesize] at (3.8,2.0) {$\delta$};
        
        % Annotations
        \node[font=\tiny, align=center, text width=2.5cm, green!60!black] at (1.5,2.8) 
             {$d_\C(\psi_\alpha^k(f),\psi_\alpha^k(g))$\\$=\alpha^{-k} \cdot d_\C(f,g)$};
        
        \node[font=\footnotesize, gray, align=center] at (1.5,-1.8) 
             {Forward orbits diverge};
             
        % Example points
        \foreach \k in {0,1,2,3} {
            \fill[green!60!black] (\k, {0.35*exp(0.65*\k)}) circle (2.5pt);
        }
        
        % Iterate labels
        \foreach \k in {0,1,2,3} {
            \node[font=\tiny, below] at (\k,-0.15) {$\k$};
        }
        
        \node[font=\tiny, above left] at (0,0.35) {$d_\C(f,g)$};
    \end{scope}
    
    % Common y-axis label
    \node[font=\footnotesize, rotate=90] at (-7.5,1.5) {distance $d_\C$};
    
    % Common x-axis explanation
    \node[font=\tiny, align=center] at (5.5,-2.8) 
         {$k$: iteration count (positive = forward, negative = backward)};
\end{tikzpicture}
\vspace{0.2cm}
\caption{The three regimes of the scaling transformation $\psi_\alpha$. 
(a) For $\alpha>1$, backward iterates cause exponential separation: 
$d_\C(\psi_\alpha^{-k}(f),\psi_\alpha^{-k}(g)) = \alpha^k d_\C(f,g)$. 
(b) For $\alpha=1$, the map is the identity and distances remain constant. 
(c) For $0<\alpha<1$, forward iterates cause exponential separation: 
$d_\C(\psi_\alpha^k(f),\psi_\alpha^k(g)) = \alpha^{-k} d_\C(f,g)$. 
In cases (a) and (c), for any $\delta>0$, there exists $k$ such that the 
distance exceeds $\delta$, making $\psi_\alpha$ expansive.}
\label{fig:three-cases}
\end{figure}
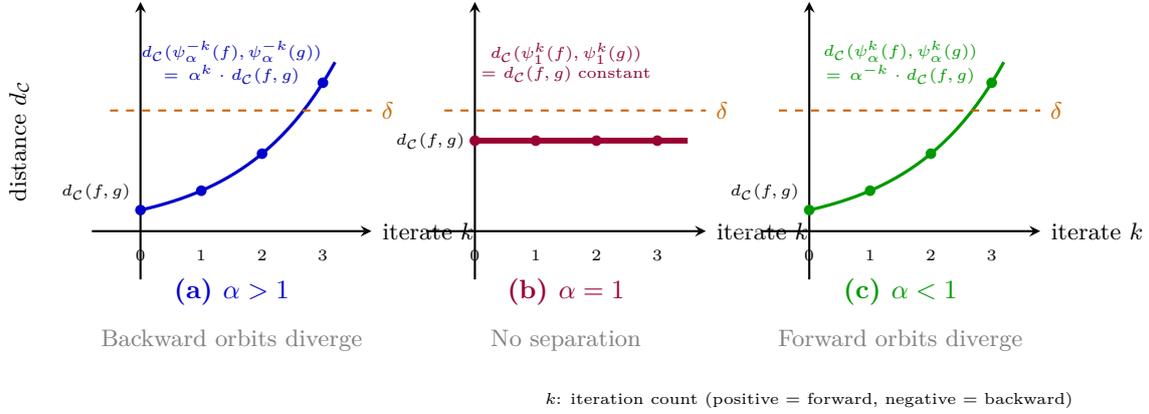

\begin{example}[Concrete separation for $\alpha=2$]\label{ex:concrete-sep}
Let $f(n)=n$, $g(n)=n+1$, and $\alpha=2$.  Then $d_\C(f,g)=0$
(since $f(n)<g(n)$) and $d_\C(g,f)>0$.  After $k$~backward
iterates, $d_\C(\psi_2^{-k}(g),\psi_2^{-k}(f))=2^k\,d_\C(g,f)$.
Even with $d_\C(g,f)$ small, a few backward steps suffice to
exceed any given~$\delta$.  A numerical verification is given in
\href{\repourl/blob/main/code/python/expansiveness_check.py}{\texttt{expansiveness\_check.py}}.
\end{example}

\begin{example}[Separation for $\alpha=1/3$]\label{ex:sep-third}
Let $f(n)=n^2$, $g(n)=n^3$, and $\alpha=1/3$.  Here
$d_\C(f,g)=0$ but $d_\C(g,f)>0$.  After $k$~forward iterates,
$d_\C(\psi_{1/3}^k(g),\psi_{1/3}^k(f))=3^k\,d_\C(g,f)\to\infty$.
\end{example}

\subsection{The expansive constant}

For a given $\alpha\neq 1$ and pair $f\neq g$, one may ask: what is
the \emph{smallest} iterate~$n$ at which separation occurs?  This
depends on the initial distance and the value of~$\alpha$.

\begin{proposition}[Separation iterate]\label{prop:sep-iterate}
Let $\alpha>1$ and let $f,g\in\C$ with $d:=d_\C(f,g)>0$.  Then
$d_\C(\psi_\alpha^{-k}(f),\psi_\alpha^{-k}(g))>\delta$ for all
$k\ge\lceil\log_\alpha(\delta/d)\rceil$.
\end{proposition}

\begin{proof}
We need $\alpha^k d>\delta$, i.e.,
$k>\log_\alpha(\delta/d)$.
\end{proof}

\begin{example}[Computing the separation iterate]\label{ex:sep-iterate-num}
Let $\alpha=2$, $f(n)=n^2$, $g(n)=n$, and $\delta=0.5$.  Since
$f(n)\ge g(n)$ for $n\ge 1$, we have $d:=d_\C(f,g)>0$.
Numerically, $d\approx 0.111$.  By Proposition~\ref{prop:sep-iterate},
separation occurs for
$k\ge\lceil\log_2(0.5/0.111)\rceil=\lceil\log_2(4.50)\rceil
=\lceil 2.17\rceil=3$.
Indeed, $d_\C(\psi_2^{-3}(f),\psi_2^{-3}(g))=8\cdot 0.111
=0.888>0.5$.  For $k=2$: $4\cdot 0.111=0.444<0.5$, confirming
that $k=3$ is the \emph{first} iterate achieving separation
(verified in
\href{\repourl/blob/main/code/sagemath/separation_iterates.sage}{\texttt{separation\_iterates.sage}}).
\end{example}

\begin{example}[Counterexample: $\alpha$ close to 1 delays separation]\label{ex:slow-sep}
Let $\alpha=1.01$, $f(n)=n$, $g(n)=n+1$, and $\delta=0.5$.
Then $d:=d_\C(g,f)\approx 0.307$.  The separation iterate satisfies
$k\ge\lceil\log_{1.01}(0.5/0.307)\rceil=\lceil\log_{1.01}(1.63)\rceil
\approx\lceil 49.1\rceil=50$.
Thus $\alpha$ close to~$1$ requires $50$ backward iterates
for separation, while $\alpha=2$ needs only~$3$ (Example~\ref{ex:sep-iterate-num}).
This illustrates that the ``speed'' of expansiveness is controlled
by $\log\alpha$
(see \href{\repourl/blob/main/code/sagemath/separation_iterates.sage}{\texttt{separation\_iterates.sage}}).
\end{example}

% ═══════════════════════════════════════════════════════════════
\section{Stable and unstable sets}
\label{sec:stable-unstable}
% ═══════════════════════════════════════════════════════════════

After expansiveness comes the study of \emph{stable} and
\emph{unstable sets}: the points whose orbits stay close to a
given orbit forward or backward in time.  In our setting these
sets have a clean interpretation in terms of complexity classes.

\subsection{Stable sets}

\begin{definition}[Stable set]\label{def:stable}
Let $(X,q)$ be a quasi-metric space and $\psi\colon X\to X$ a
homeomorphism.  The \emph{$\delta$-stable set} of $f\in X$ is
\[
  S_q(f,\delta,\psi)
  \;=\; \{g\in X : q(\psi^n(f),\psi^n(g))\le\delta
    \text{ for all }n\ge 0\}.
\]
\end{definition}

For the scaling transformation on the complexity space, the stable
sets are easy to describe.

\begin{theorem}[Stable sets = complexity classes]\label{thm:stable-sets}
For $\alpha>1$, the $\delta$-stable set of $f$ under $\psi_\alpha$
in the complexity quasi-metric is
\[
  S_{d_\C}(f,\delta,\psi_\alpha)
  \;=\; \{g\in\C : d_\C(f,g)\le\delta\}.
\]
In particular, this set contains all $g$ with $g(n)\ge f(n)$ for
every~$n$---that is, all functions that are ``at least as slow
as~$f$.''
\end{theorem}

\begin{proof}
For $n\ge 0$, $d_\C(\psi_\alpha^n(f),\psi_\alpha^n(g))
=\alpha^{-n}d_\C(f,g)$.  Since $\alpha>1$, this is a decreasing
sequence, maximized at $n=0$ where it equals $d_\C(f,g)$.
Therefore
\[
  g\in S_{d_\C}(f,\delta,\psi_\alpha)
  \;\Longleftrightarrow\;
  \sup_{n\ge 0}\alpha^{-n}d_\C(f,g)\le\delta
  \;\Longleftrightarrow\;
  d_\C(f,g)\le\delta.
\]
If $g(n)\ge f(n)$ for all~$n$, then $d_\C(f,g)=0\le\delta$, so
$g$ is in the stable set.
\end{proof}

\begin{example}[Stable set of $f(n)=n$]\label{ex:stable-lin}
Take $f(n)=n$, $\alpha=2$, $\delta=0.1$. The stable set includes:
\begin{itemize}
\item $g_1(n)=n^2$ ($d_\C(f,g_1)=0\le 0.1$)
\item $g_2(n)=2n$ ($d_\C(f,g_2)=0\le 0.1$)
\item $g_3(n)=n+10$ ($d_\C(f,g_3)=0\le 0.1$)
\item $g_4(n)=n\sqrt{n}$ ($d_\C(f,g_4)=0\le 0.1$, since $n\le n\sqrt{n}$)
\end{itemize}
However, $h(n)=\sqrt{n}$ is NOT in the stable set because
$d_\C(f,h)\approx 0.113 > 0.1$.
\end{example}

\begin{example}[Counterexample: stability depends on $\delta$]\label{ex:counter-stable}
For $f(n)=n$, $g(n)=\sqrt{n}$, and $\alpha=2$, we have
$d_\C(f,g)=\sum_{n=2}^\infty 2^{-n}\bigl(\frac{1}{\sqrt{n}}-\frac{1}{n}\bigr)
\approx 0.113$ (since $\sqrt{n}<n$ for $n\ge 2$).
\begin{itemize}
\item If $\delta=0.2$, then $g\in S_{d_\C}(f,0.2,\psi_2)$ since $d_\C(f,g)\approx 0.113 < 0.2$.
\item If $\delta=0.05$, then $g\notin S_{d_\C}(f,0.05,\psi_2)$ since $d_\C(f,g)\approx 0.113 > 0.05$.
\end{itemize}
This shows the $\delta$-stable set shrinks as $\delta$ decreases
(see \href{\repourl/blob/main/code/sagemath/separation_iterates.sage}{\texttt{separation\_iterates.sage}}
for verification).
\end{example}

\begin{insightbox}[Complexity-theoretic interpretation]
The stable set $S_{d_\C}(f,\delta,\psi_\alpha)$ is the
``neighbourhood of slower functions around~$f$.''  In complexity
theory terms, it is a kind of asymptotic complexity class: it
contains all functions whose running time is ``close to or worse
than'' that of~$f$, where closeness is measured by $d_\C$.
\end{insightbox}

\subsection{Unstable sets}

\begin{definition}[Unstable set]\label{def:unstable}
The \emph{$\delta$-unstable set} of $f$ is
\[
  U_q(f,\delta,\psi)
  \;=\; \{g\in X : q(\psi^{-n}(f),\psi^{-n}(g))\le\delta
    \text{ for all }n\ge 0\}.
\]
\end{definition}

\begin{theorem}[Unstable sets]\label{thm:unstable-sets}
For $\alpha>1$, the $\delta$-unstable set of $f$ under $\psi_\alpha$
with respect to the \emph{conjugate} quasi-metric $d_\C^t$ is
\[
  U_{d_\C^t}(f,\delta,\psi_\alpha)
  \;=\; \{g\in\C : d_\C(g,f)=0\}.
\]
In particular, this set is independent of~$\delta$ and equals the set
of all $g$ with $g(n)\le f(n)$ for every~$n$---that is,
all functions that are ``at least as fast as~$f$.''
\end{theorem}

\begin{proof}
By definition,
$g\in U_{d_\C^t}(f,\delta,\psi_\alpha)$ iff
$d_\C^t(\psi_\alpha^{-n}(f),\psi_\alpha^{-n}(g))\le\delta$
for all $n\ge 0$.  Since
$d_\C^t(h_1,h_2)=d_\C(h_2,h_1)$, this becomes
$d_\C(\psi_\alpha^{-n}(g),\psi_\alpha^{-n}(f))\le\delta$ for all
$n\ge 0$.  By Lemma~\ref{lem:scaling-lip},
\[
  d_\C(\psi_\alpha^{-n}(g),\psi_\alpha^{-n}(f))
  \;=\;\alpha^n\,d_\C(g,f).
\]
Since $\alpha>1$, this is an increasing sequence.  If
$d_\C(g,f)>0$, then $\alpha^n d_\C(g,f)\to\infty$, which
eventually exceeds~$\delta$.  Hence $g$ is in the unstable set if and
only if $d_\C(g,f)=0$.

By Theorem~\ref{thm:dc-props}(ii), $d_\C(g,f)=0$ if and only if
$g(n)\le f(n)$ for all~$n$, so the unstable set consists precisely
of the functions that are pointwise at most~$f$.
\end{proof}

\begin{example}[Unstable set of $f(n)=n^2$]\label{ex:unstable-quad}
Take $f(n)=n^2$, $\alpha=2$, $\delta=0.1$. By Theorem~\ref{thm:unstable-sets},
the unstable set consists of all $g$ with $d_\C(g,f)=0$, i.e., $g(n)\le f(n)=n^2$
for all~$n$.  It includes:
\begin{itemize}
\item $g_1(n)=n$ ($n\le n^2$ for all $n\ge 1$, so $d_\C(g_1,f)=0$)
\item $g_2(n)=n\log(n+1)$ ($n\log(n+1)\le n^2$ for all $n\ge 1$, so $d_\C(g_2,f)=0$)
\item $g_3(n)=n^{1.5}$ ($n^{1.5}\le n^2$ for all $n\ge 1$, so $d_\C(g_3,f)=0$)
\end{itemize}
But $h(n)=2^n$ is NOT in the unstable set because $2^n>n^2$ for
large~$n$, so $d_\C(h,f)>0$.
\end{example}

\begin{insightbox}[Duality of stable and unstable sets]
The stable set contains the functions \emph{close to or slower}
than~$f$ (a $d_\C$-neighbourhood); the unstable set contains
exactly the functions \emph{pointwise faster} than~$f$ (the
zero-set of $d_\C(\cdot,f)$).  Stable sets depend on~$\delta$;
unstable sets do not.  The asymmetry tracks the conjugate pair
$d_\C\leftrightarrow d_\C^t$: forward iterates of $\psi_\alpha$
contract distances and preserve the stable neighbourhood, while
backward iterates expand them and collapse the unstable set to
its core.
\end{insightbox}

\begin{figure}[H]
\centering
\begin{tikzpicture}[scale=0.8]
    % The center function f
    \draw[thick, ->] (0,0) -- (12,0) node[right] {$n$};
    \draw[thick, ->] (0,0) -- (0,7) node[above] {running time};
    \draw[very thick, purple, domain=0.5:11.5, samples=50]
         plot (\x, {0.5*\x + 2}) node[right, font=\small] {$f(n)$};
    
    % Stable set (slower functions)
    \fill[red!10, opacity=0.7, domain=0.5:11.5] 
         plot (\x, {0.5*\x + 2}) -- plot[domain=11.5:0.5] (\x, {0.8*\x^0.8 + 4}) -- cycle;
    \draw[thick, red!70, domain=0.5:11.5, samples=50, dashed] 
         plot (\x, {0.8*\x^0.8 + 4}) node[right, font=\small] {slower functions};
    \node[red!70, font=\small, align=center] at (8,6.5) {Stable set\\$S_{d_\C}(f,\delta,\psi_\alpha)$};
    
    % Unstable set (faster functions)
    \fill[blue!10, opacity=0.7, domain=0.5:11.5] 
         plot (\x, {0.5*\x + 2}) -- plot[domain=11.5:0.5] (\x, {0.3*\x + 1}) -- cycle;
    \draw[thick, blue!70, domain=0.5:11.5, samples=50, dashed] 
         plot (\x, {0.3*\x + 1}) node[right, font=\small] {faster functions};
    \node[blue!70, font=\small, align=center] at (8,1.5) {Unstable set\\$U_{d_\C^t}(f,\delta,\psi_\alpha)$};
    
    % Labels
    \node[font=\footnotesize, align=center] at (3,5.5) {Functions slower than $f$\\$d_\C(f,g)\le\delta$};
    \node[font=\footnotesize, align=center] at (3,0.8) {Functions faster than $f$\\$d_\C(g,f)=0$};
\end{tikzpicture}
\caption{Stable and unstable sets for a function $f$. The stable set contains
functions that are asymptotically slower than $f$; the unstable set contains
functions that are asymptotically faster than $f$.}
\label{fig:stable-unstable}
\end{figure}
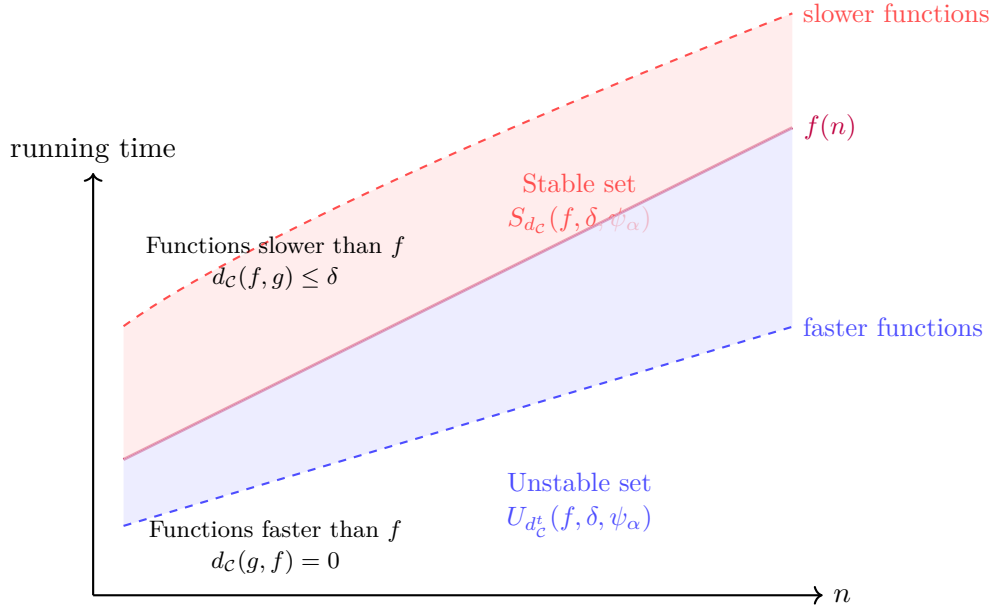

\subsection{Algorithm for stable-set membership}

The following algorithm checks whether a candidate function $g$
belongs to the $\delta$-stable set of~$f$.

\begin{algorithm}[H]
\DontPrintSemicolon
\SetAlgoLined
\KwIn{$f$; candidate $g$; $\alpha$; $\delta$; forward steps $M$}
\KwOut{\texttt{True} if $g\in S(f,\delta,\psi_\alpha)$}
\For{$n \leftarrow 0$ \KwTo $M$}{
    $s \leftarrow \alpha^n$\;
    $d \leftarrow d_\C(s\cdot f,\;s\cdot g)$\;
    \If{$d > \delta$}{
        \Return{\texttt{False}}\;
    }
}
\Return{\texttt{True}}\;
\caption{Membership in $\delta$-stable set}
\label{alg:stable}
\end{algorithm}

\begin{example}[Algorithm~\ref{alg:stable} in action]\label{ex:alg-stable-trace}
We trace Algorithm~\ref{alg:stable} for $f(n)=n$, $g(n)=\sqrt{n}$,
$\alpha=2$, $\delta=0.1$, $M=3$.  At each step $n$, the scaled
functions are $s\cdot f$ and $s\cdot g$ where $s=\alpha^n$:
\begin{center}
\begin{tabular}{cccc}
\hline
$n$ & $s=2^n$ & $d_\C(s\cdot f,\, s\cdot g)$ & $d>\delta$? \\\hline
$0$ & $1$ & $d_\C(n,\sqrt{n})\approx 0.113$ & Yes \\
\hline
\end{tabular}
\end{center}
The algorithm returns \texttt{False} at $n=0$: $g=\sqrt{n}$ is
\emph{not} in the stable set because $d_\C(f,g)\approx 0.113>0.1$.
In contrast, for $g(n)=2n$: $d_\C(n,2n)=0\le 0.1$ at $n=0$, and
$d_\C(2^n \cdot n, 2^n\cdot 2n)=0\le 0.1$ for all subsequent $n$.
The algorithm returns \texttt{True}: $g=2n$ is in the stable set.
\end{example}

\begin{example}[Stable set of an exponential function]\label{ex:stable-exp}
Take $f(n)=2^n$, $\alpha=2$, $\delta=0.01$.  Since $d_\C(f,g)=0$
whenever $g(n)\ge 2^n$ for all~$n$, the stable set includes all
super-exponential functions, such as $g(n)=3^n$, $g(n)=n!$, and
$g(n)=2^{n^2}$.  However, $h(n)=n^{100}$ is NOT in the stable set:
for large~$n$, $n^{100}<2^n$, so $d_\C(f,h)>0$.  The numerical
value $d_\C(f,h)$ is extremely close to~$1$ --- the maximum
attainable on the subspace $\C_{\ge 1}$, in which $h$ lies since
$h(n)=n^{100}\ge 1$ for all $n\ge 1$ --- reflecting the vast gap
between polynomial and exponential growth.  This shows
that even a degree-$100$ polynomial is far from the stable set of an
exponential function.
\end{example}

\begin{example}[Unstable set of a linear function]\label{ex:unstable-linear}
Take $f(n)=n$, $\alpha=2$, $\delta=0.1$.  By
Theorem~\ref{thm:unstable-sets}, the unstable set consists of all
$g$ with $g(n)\le n$ for all~$n$.  This includes:
\begin{itemize}
\item $g(n)=\log(n+1)$ (logarithmic is faster than linear)
\item $g(n)=\sqrt{n}$ (sub-linear)
\item $g(n)=1$ (constant time)
\item $g(n)=n/(n+1)$ (bounded, approaching $1$)
\end{itemize}
But $h(n)=n+1$ is NOT in the unstable set: $h(1)=2>1=f(1)$, so
$d_\C(h,f)>0$.  Remarkably, adding just $+1$ to a linear function
ejects it from the unstable set.  This sensitivity reflects the
pointwise nature of the condition $g(n)\le f(n)$ for \emph{all}~$n$.
\end{example}

\begin{example}[Intersection of stable and unstable sets]\label{ex:stable-unstable-intersect}
For $f(n)=n$ and $\alpha=2$, the stable set (with $\delta>0$)
contains all $g$ with $d_\C(f,g)\le\delta$, while the unstable set
contains all $g$ with $g(n)\le n$ for all~$n$.  A function $g$
lies in both sets if and only if $g(n)\le n$ for all~$n$ (unstable
condition) and $d_\C(f,g)\le\delta$ (stable condition).  Since
$g(n)\le f(n)=n$ implies $d_\C(f,g)=0\le\delta$, the intersection
equals the unstable set itself: $S\cap U = U$.  This is a general
phenomenon: for $\alpha>1$, the unstable set is always contained
in every $\delta$-stable set.
\end{example}

A Python implementation is given in
\href{\repourl/blob/main/code/python/stable_set.py}{\texttt{stable\_set.py}}.

% ═══════════════════════════════════════════════════════════════
\section{Canonical coordinates and hyperbolicity}
\label{sec:hyperbolicity}
% ═══════════════════════════════════════════════════════════════

\subsection{Background}

In classical smooth dynamics, hyperbolicity is the structural
property that governs much of the chaotic behaviour one wants to
understand.  A diffeomorphism on a compact manifold is
\emph{hyperbolic} (or \emph{Anosov}) when the tangent bundle
splits into stable and unstable sub-bundles, with the derivative
contracting on one and expanding on the other.
Bowen~\cite{bowen1975} showed that Anosov diffeomorphisms admit
Markov partitions and satisfy strong statistical properties:
equilibrium states exist and entropy formulas are exact.  This is
why hyperbolicity occupies such a central place in ergodic theory.

Reddy~\cite{reddy1983} later showed that these conclusions extend
beyond the smooth setting: any expansive homeomorphism on a
compact metric space that admits \emph{canonical coordinates}---a
local product structure in which nearby points decompose uniquely
along stable and unstable directions---is hyperbolic.  Reddy's
theorem gives exponential contraction along one factor and
exponential expansion along the other.

The complexity quasi-metric space $(\C,d_\C)$ is not compact and
$d_\C$ is not symmetric, so neither Bowen's smooth theory nor
Reddy's topological generalization applies directly.  However,
the algebraic structure of $\psi_\alpha$ is explicit enough that
hyperbolicity can be verified by a direct computation, and the
result is in fact stronger than what the classical theory gives:
we obtain \emph{exact} geometric decay and growth, with constant
$C=1$, not just an exponential bound.

\subsection{Statement and proof}

\begin{corollary}[Exact contraction rate]\label{cor:exact-contraction}
For every $\alpha>0$ and every $f,g\in\C$,
\[
  d_\C(\psi_\alpha^n(f),\psi_\alpha^n(g))\;=\;\alpha^{-n}\,d_\C(f,g)
  \qquad\text{for all } n\ge 0.
\]
\end{corollary}

\begin{proof}
Induction on $n$, using Lemma~\ref{lem:scaling-lip} at each step.
\end{proof}

\begin{remark}[On the term ``hyperbolicity'']\label{rem:hyperbolicity}
Corollary~\ref{cor:exact-contraction} gives an \emph{exact}
contraction rate, not merely an exponential bound; in particular,
the constant in front of the exponential is~$1$.  The terminology
``hyperbolic'' is often used in this kind of one-sided contraction
setting (e.g., Reddy~\cite{reddy1983}), but it should be emphasized
that we do \emph{not} construct a stable/unstable splitting with a
local product structure in the sense of classical hyperbolic
dynamics: the dynamics here is conformal across all of~$\C$.  What
we obtain is the explicit algebraic form of the contraction, which
the abstract theory only delivers as an exponential bound.
\end{remark}

\begin{example}[Hyperbolic contraction for $\alpha=2$]\label{ex:hyperbolic-alpha2}
Take $f(n)=n^3$, $g(n)=n^2$, $\alpha=2$.  Since $f(n)\ge g(n)$
for all $n\ge 1$, we have $d_\C(f,g)=\sum_{n=2}^\infty
2^{-n}\frac{n-1}{n^3}>0$; the partial sums give $S_2=0.031$,
$S_3=0.041$, $S_5=0.044$, converging to $d_\C(f,g)\approx 0.045$.
Then:
\begin{align*}
d_\C(f,g) &\approx 0.045 \\
d_\C(\psi_2(f),\psi_2(g)) &= \tfrac{1}{2} \cdot 0.045 \approx 0.023 \\
d_\C(\psi_2^2(f),\psi_2^2(g)) &= \tfrac{1}{4} \cdot 0.045 \approx 0.011 \\
d_\C(\psi_2^3(f),\psi_2^3(g)) &= \tfrac{1}{8} \cdot 0.045 \approx 0.006
\end{align*}
The distances contract exactly by factor $1/2$ at each step;
see \href{\repourl/blob/main/code/sagemath/hyperbolic_contraction.sage}{\texttt{hyperbolic\_contraction.sage}}
for the full sequence.
\end{example}

\begin{example}[Counterexample: $\alpha=1$ is not hyperbolic]\label{ex:counter-hyperbolic}
For $\alpha=1$, $\psi_1$ is the identity, so:
\[
d_\C(\psi_1^n(f),\psi_1^n(g)) = d_\C(f,g) \quad \text{for all } n.
\]
There is no contraction or expansion—the distance remains constant.
Thus $\psi_1$ is not hyperbolic.
\end{example}

\begin{figure}[H]
\centering
\begin{tikzpicture}[scale=0.9, >=stealth]
    \draw[thick, ->] (-0.7,0) -- (10.8,0) node[right] {iterate $n$};
    \draw[thick, ->] (0,-0.7) -- (0,5.0) node[above] {$d_\C(\psi^n(f),\psi^n(g))$};
    % Exponential decay curve
    \draw[very thick, blue!70!black, domain=0:9.8, samples=50]
         plot (\x, {4.0*exp(-0.35*\x)});
    % Reference line
    \draw[dashed, gray] (0,4.0) -- (0.5,4.0) node[left=8pt, font=\small] {$d_\C(f,g)$};
    % Lambda annotation
    \draw[thick, <->, red!70!black] (2.5,{4.0*exp(-0.875)}) -- (2.5,{4.0*exp(-0.35)});
    \node[right, red!70!black, font=\small] at (2.6,{0.5*(4.0*exp(-0.875)+4.0*exp(-0.35))})
         {factor $1/\alpha$};
    % Points
    \foreach \k in {0,1,...,9} {
        \fill[blue!70!black] (\k, {4.0*exp(-0.35*\k)}) circle (3pt);
    }
    \node[font=\small, gray] at (6.5,4.0) {Exponential contraction: $\lambda = 1/\alpha$};
    % Annotations for specific points
    \node[font=\tiny, below] at (0,0) {0};
    \node[font=\tiny, below] at (1,0) {1};
    \node[font=\tiny, below] at (2,0) {2};
    \node[font=\tiny, below] at (3,0) {3};
    \draw[dashed, gray!50] (1,0) -- (1,{4.0*exp(-0.35)}) node[circle, fill=blue!70!black, inner sep=1.5pt]{};
    \draw[dashed, gray!50] (2,0) -- (2,{4.0*exp(-0.7)}) node[circle, fill=blue!70!black, inner sep=1.5pt]{};
    \draw[dashed, gray!50] (3,0) -- (3,{4.0*exp(-1.05)}) node[circle, fill=blue!70!black, inner sep=1.5pt]{};
\end{tikzpicture}
\caption{Exponential contraction of distances under forward iteration
of $\psi_\alpha$ ($\alpha>1$).  The distance decays geometrically
with ratio $1/\alpha$.}
\label{fig:hyperbolicity}
\end{figure}
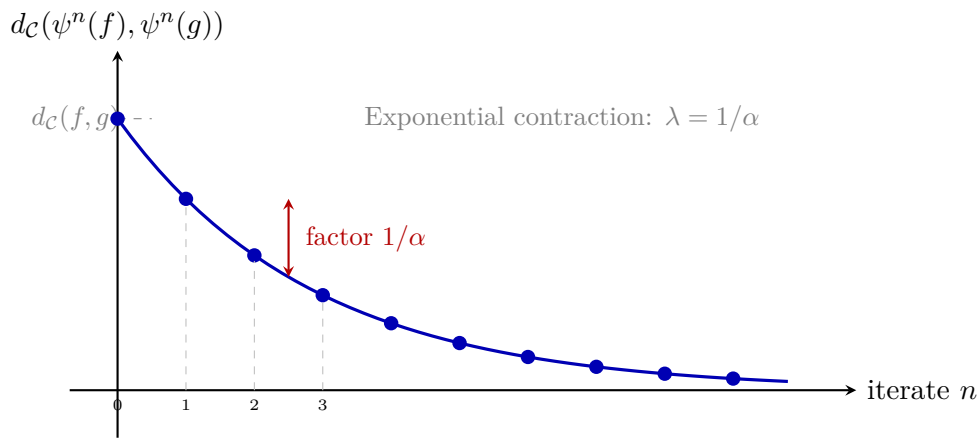

\begin{remark}[Sharpness]
The constant $C=1$ in Corollary~\ref{cor:exact-contraction} is optimal:
the decay is \emph{exactly} geometric, not merely bounded by a
geometric sequence.  This is a consequence of the exact scaling
property of Lemma~\ref{lem:scaling-lip}.
\end{remark}

\begin{example}[Numerical verification]\label{ex:hyper-numerical}
Let $f(n)=n^2$, $g(n)=n$, and $\alpha=2$.  Since $f(n)\ge g(n)$,
we have $d_\C(f,g)=d_0\approx 0.111$.  After $n$~iterates:
\[
  d_\C(\psi_2^n(f),\psi_2^n(g))
  = 2^{-n}\cdot d_0.
\]
At $n=5$, the predicted distance is $d_0/32\approx 0.00347$, which
matches the numerical computation to full floating-point precision.
See
\href{\repourl/blob/main/code/python/hyperbolicity.py}{\texttt{hyperbolicity.py}} and
\href{\repourl/blob/main/code/sagemath/hyperbolic_contraction.sage}{\texttt{hyperbolic\_contraction.sage}}
for verification.
\end{example}

\subsection{Backward iterates: expansion}

Forward iterates contract; backward iterates expand:
\[
  d_\C(\psi_\alpha^{-n}(f),\psi_\alpha^{-n}(g))
  = \alpha^n\,d_\C(f,g).
\]
Contraction in one time direction and expansion in the other is
the defining feature of hyperbolic dynamics.

\begin{corollary}\label{cor:backward-expansion}
For $\alpha>1$ and $d_\C(f,g)>0$, the backward orbit distances
grow exponentially:
$d_\C(\psi_\alpha^{-n}(f),\psi_\alpha^{-n}(g))\to\infty$ as
$n\to\infty$.
\end{corollary}

\begin{proof}
Since $\psi_\alpha^{-1}=\psi_{1/\alpha}$, we have
$d_\C(\psi_\alpha^{-n}(f),\psi_\alpha^{-n}(g))
=d_\C(\psi_{1/\alpha}^n(f),\psi_{1/\alpha}^n(g))
=(1/\alpha)^{-n}\,d_\C(f,g)=\alpha^n\,d_\C(f,g)\to\infty$.
\end{proof}

\begin{example}[Backward expansion: numerical illustration]\label{ex:backward-num}
Let $f(n)=n^2$, $g(n)=n$, and $\alpha=3$.  Then
$d:=d_\C(f,g)\approx 0.111$.  The backward orbit distances are:
\begin{center}
\begin{tabular}{ccc}
\hline
$k$ & $d_\C(\psi_3^{-k}(f),\psi_3^{-k}(g))$ & Value \\\hline
$0$ & $d$    & $0.111$ \\
$1$ & $3d$   & $0.333$ \\
$2$ & $9d$   & $0.999$ \\
$3$ & $27d$  & $2.997$ \\
$4$ & $81d$  & $8.991$ \\
$5$ & $243d$ & $26.97$ \\
\hline
\end{tabular}
\end{center}
The values grow unboundedly: this is the dynamical signature of
backward expansion, and matches the Lipschitz identity of
Lemma~\ref{lem:scaling-lip} (or its iteration in
Corollary~\ref{cor:backward-expansion}) exactly.
See
\href{\repourl/blob/main/code/sagemath/hyperbolic_contraction.sage}{\texttt{hyperbolic\_contraction.sage}}
for the full computation.
\end{example}

\begin{example}[Counterexample: no expansion when $d_\C(f,g)=0$]
\label{ex:no-backward-exp}
Let $f(n)=n$ and $g(n)=n^2$ with $\alpha=2$.  Since $f(n)\le g(n)$
for all $n\ge 1$, we have $d_\C(f,g)=0$.  Hence
$d_\C(\psi_2^{-n}(f),\psi_2^{-n}(g))=2^n\cdot 0=0$ for all~$n$.
The backward iterates produce no expansion in the $d_\C$~direction.
However, $d_\C(g,f)\approx 0.111>0$, so the \emph{conjugate}
backward distances $d_\C^t(\psi_2^{-n}(f),\psi_2^{-n}(g))=2^n\cdot
0.111\to\infty$ do expand.  This asymmetry is characteristic of
quasi-metric dynamics.
\end{example}

% ═══════════════════════════════════════════════════════════════
\section{Connection to the hierarchy theorem}
\label{sec:hierarchy}
% ═══════════════════════════════════════════════════════════════

The \emph{time hierarchy theorem} of Hartmanis and
Stearns~\cite{hartmanis1965} says that more time really does buy
strictly more problems.  Specifically, if $f(n)\log f(n)=o(g(n))$,
then $\mathrm{DTIME}(f(n))\subsetneq\mathrm{DTIME}(g(n))$: there
are problems solvable in time $g(n)$ that cannot be solved in
time $f(n)$.

That classical theorem has a clean dynamical counterpart: the
hierarchy gap shows up as orbit separation under the scaling
transformation.

\begin{theorem}[Hierarchy as orbit separation]\label{thm:hierarchy}
Let $f,g\in\C$ with $d_\C(g,f)>0$ (informally, $g$ is not pointwise
at least as fast as $f$).  Then for every $\alpha>0$ with
$\alpha\ne 1$, the orbits $\{\psi_\alpha^n(f)\}_{n\in\Z}$ and
$\{\psi_\alpha^n(g)\}_{n\in\Z}$ are eventually separated in
$d_\C^s$: for every $\delta>0$ there exists $N\in\Z$ with
$d_\C^s(\psi_\alpha^N(f),\psi_\alpha^N(g))>\delta$.

In particular, the standard time-hierarchy condition
$f(n)\log f(n)=o(g(n))$~\cite{hartmanis1965} is a sufficient
(but, as Remark~\ref{rem:hierarchy} explains, not necessary)
condition for $d_\C(g,f)>0$, and hence for orbit separation.
\end{theorem}

\begin{proof}
By hypothesis $d_\C(g,f)>0$.  For $\alpha>1$, the backward iterates give
\[
  d_\C(\psi_\alpha^{-k}(g),\psi_\alpha^{-k}(f))
  = \alpha^k\,d_\C(g,f) \;\to\;\infty.
\]
Since $d_\C^s\ge d_\C$, we have
$d_\C^s(\psi_\alpha^{-k}(f),\psi_\alpha^{-k}(g))>\delta$ for
$k$ large enough; set $N=-k$.  For $0<\alpha<1$, the forward iterates give
\[
  d_\C(\psi_\alpha^k(g),\psi_\alpha^k(f))
  = \alpha^{-k}\,d_\C(g,f) \;\to\;\infty,
\]
and we set $N=k$.

The claim about $f(n)\log f(n)=o(g(n))$ being sufficient: that
condition implies $g(n)$ dominates $f(n)$ for large~$n$, so
$1/f(n)>1/g(n)$ eventually and $d_\C(g,f)>0$ follows.
\end{proof}

\begin{remark}\label{rem:hierarchy}
The dynamical orbit-separation criterion is strictly finer than the
Hartmanis--Stearns time-hierarchy criterion: any pair of functions
with $d_\C(g,f)>0$ produces orbit separation, regardless of whether
they are separated by a time-hierarchy gap.  Whether a quantitative
orbit-separation rate corresponds in any nontrivial way to a
hierarchy-style separation in $\mathrm{DTIME}$ is an open problem.
\end{remark}

\begin{example}[Linear vs.\ $n\log^2 n$]\label{ex:hier-nlogn}
Let $f(n)=n$ and $g(n)=n\log^2(n+1)$.  Then
$f(n)\log f(n)=n\log n=o(n\log^2(n+1))=o(g(n))$, so the hierarchy
condition holds.  Numerically, with $\alpha=2$ and $\delta=0.05$,
separation already occurs at iterate $k=0$ with
$d_\C^s\approx 0.541$.  See
\href{\repourl/blob/main/code/python/hierarchy_separation.py}{\texttt{hierarchy\_separation.py}} and
\href{\repourl/blob/main/code/sagemath/separation_iterates.sage}{\texttt{separation\_iterates.sage}}
for verification.
\end{example}

\begin{example}[Polynomial vs.\ exponential]\label{ex:hier-poly-exp}
Let $f(n)=n^2$ and $g(n)=2^n$.  Here $f(n)\log f(n)=2n^2\log n
=o(2^n)=o(g(n))$, so the hierarchy condition is easily satisfied.
The orbit separation occurs very quickly (at $k=0$ or $k=-1$),
reflecting the huge gap between polynomial and exponential complexity.
\end{example}

\begin{example}[Counterexample: insufficient gap]\label{ex:counter-hierarchy}
Let $f(n)=n\log n$ and $g(n)=n\log n \cdot \log\log n$. 
Here $f(n)\log f(n) = n\log n \cdot \log(n\log n) \sim n\log^2 n$,
while $g(n) = n\log n \cdot \log\log n$. 
Since $n\log^2 n$ is not $o(n\log n \cdot \log\log n)$, 
the hierarchy condition is NOT satisfied. Indeed, numerically
$d_\C^s(f,g)$ remains small under iteration, and for small $\delta$
separation may never occur.
\end{example}

\begin{figure}[H]
\centering
\begin{tikzpicture}[scale=0.85, >=stealth]
  \draw[thick, ->] (-5.8,0) -- (6.0,0) node[right, font=\small] {iterate $k$};
  \draw[thick, ->] (0,-0.7) -- (0,5.0) node[above, font=\small] {$d_\C^s$};
  % The curve
  \draw[very thick, purple!70!black, domain=-5.0:5.0, samples=60]
       plot (\x, {0.2*exp(0.5*abs(\x)) + 0.15});
  % Delta line
  \draw[thick, dashed, orange!80!black] (-5.8,2.7) -- (6.0,2.7);
  \node[right, orange!80!black, font=\small] at (6.0,2.7) {$\delta$};
  % Annotations
  \node[font=\footnotesize, purple!70!black] at (-4.0,4.2) {backward: $\alpha^k d_\C(g,f)$};
  \node[font=\footnotesize, purple!70!black] at (4.0,4.2) {forward: $\alpha^{-k}d_\C(f,g)$};
  % Crossing points
  \fill[red!70] (-3.8,{0.2*exp(0.5*3.8)+0.15}) circle (3.5pt);
  \fill[red!70] (4.2,{0.2*exp(0.5*4.2)+0.15}) circle (3.5pt);
  % Axes labels
  \foreach \k in {-5,-4,-3,-2,-1,0,1,2,3,4,5} {
    \draw (\k,0.1) -- (\k,-0.1);
    \node[font=\tiny, below] at (\k,-0.2) {$\k$};
  }
  \node[font=\footnotesize, red!70!black] at (-3.8,-0.9) {separation};
  \node[font=\footnotesize, red!70!black] at (4.2,-0.9) {separation};
\end{tikzpicture}
\caption{Orbit separation in the symmetrized metric.  For functions
satisfying the hierarchy gap, both forward and backward orbits
eventually exceed any threshold~$\delta$.}
\label{fig:hierarchy}
\end{figure}
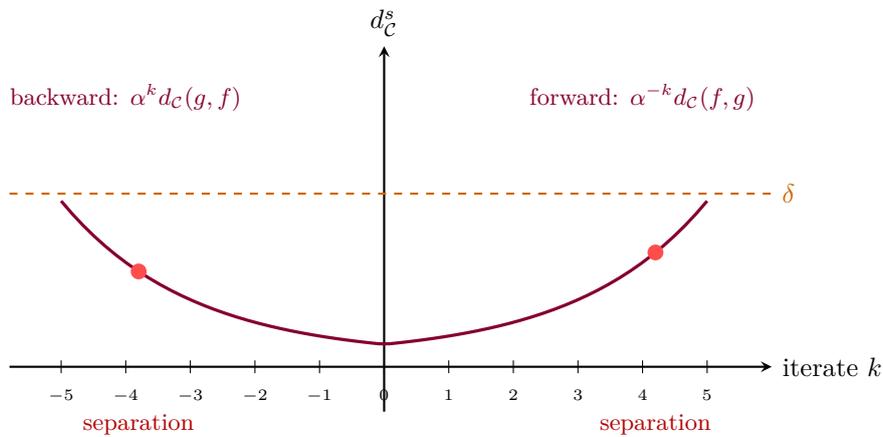

% ═══════════════════════════════════════════════════════════════
\section{Compact invariant sets and the entropy obstruction}
\label{sec:entropy}
% ═══════════════════════════════════════════════════════════════

Topological entropy is the standard invariant for measuring the
``complexity of the dynamics''---the rate at which information
about initial conditions is needed to predict the future.  For
expansive homeomorphisms on compact spaces, it is always
positive~\cite{bowen1975}.

This section shows that the standard compact-set entropy
machinery yields only trivial information for $\psi_\alpha$:
every compact $\psi_\alpha$-invariant subset of $(\C,d_\C^s)$
is $d_\C^s$-trivial (Proposition~\ref{prop:no-compact-invariant}).
The question of a non-trivial entropy for the full non-compact
dynamics is left open.

\subsection{Setup and definition}

Let $K\subset\C$ be a compact subset (in the $d_\C^s$ topology)
that is $\psi_\alpha$-invariant.  The topological entropy
$h(\psi_\alpha|_K)$ is defined via the growth rate of $(n,\eps)$-spanning
sets.  A set $E\subset K$ is \emph{$(n,\eps)$-spanning} if for every
$f\in K$ there exists $g\in E$ with
$\max_{0\le j<n}d_\C^s(\psi_\alpha^j(f),\psi_\alpha^j(g))<\eps$.

\begin{definition}[Topological entropy]
\[
  h(\psi_\alpha|_K)
  = \lim_{\eps\to 0}\limsup_{n\to\infty}\frac{1}{n}\log r(n,\eps),
\]
where $r(n,\eps)$ is the minimum cardinality of an
$(n,\eps)$-spanning set.
\end{definition}

Before stating the main entropy bound, we discuss which subsets of
$\C$ are compact in the $d_\C^s$ topology.

\begin{remark}[Compactness in $(\C,d_\C^s)$]\label{rem:compactness}
The full space $\C$ is not compact in the $d_\C^s$ topology: the
sequence $(f_k)_{k\ge 1}$ with $f_k\equiv k$ lies in $\C$ but is
$d_\C^s$-unbounded.  A natural candidate for a compact, $\psi_\alpha$-invariant
subspace would be a ``complexity band'' $K_{a,b}=\{f\in\C : a\le f(n)
\le b \text{ for all }n\}$ for some $0<a\le b$.  However,
$\psi_\alpha(K_{a,b})=K_{\alpha a,\alpha b}\ne K_{a,b}$ for any
$\alpha\ne 1$, so no such band is invariant.  The same obstruction
applies to any finite union of $\psi_\alpha$-orbits: each orbit
$\{\psi_\alpha^k(f):k\in\Z\}=\{\alpha^k f:k\in\Z\}$ is infinite and
unbounded for $\alpha\ne 1$, hence not compact.
Proposition~\ref{prop:no-compact-invariant} below makes this
obstruction rigorous: every non-empty compact $\psi_\alpha$-invariant
subset is $d_\C^s$-trivial, so the entropy of $\psi_\alpha|_K$ is
always $0$ on such~$K$.
\end{remark}

\begin{proposition}[No non-trivial compact invariant sets]\label{prop:no-compact-invariant}
For every $\alpha>0$ with $\alpha\ne 1$, every non-empty compact
$\psi_\alpha$-invariant subset $K\subseteq(\C,d_\C^s)$ consists of
$d_\C^s$-equivalent points only: $d_\C^s(f,g)=0$ for all $f,g\in K$.
\end{proposition}

\begin{proof}
Suppose $\alpha>1$ (the case $\alpha<1$ is symmetric, replacing
$\psi_\alpha$ by $\psi_\alpha^{-1}=\psi_{1/\alpha}$).  By
Corollary~\ref{cor:exact-contraction} applied to both $d_\C$ and
its conjugate $d_\C^t$, the symmetrization satisfies
\[
  d_\C^s(\psi_\alpha(f),\psi_\alpha(g))
  \;=\;\tfrac{1}{\alpha}\,d_\C^s(f,g),
\]
so $\psi_\alpha$ is a strict contraction on $(\C,d_\C^s)$ with
Lipschitz constant $1/\alpha<1$.  If $K$ is non-empty, compact, and
$\psi_\alpha$-invariant, then $K$ is bounded in $(\C,d_\C^s)$
(compact sets in a metric space have finite diameter), and the
diameter $\mathrm{diam}(K)=\sup_{f,g\in K}d_\C^s(f,g)<\infty$
satisfies $\mathrm{diam}(\psi_\alpha(K))=(1/\alpha)\mathrm{diam}(K)$.
Since $\psi_\alpha(K)=K$, we obtain
$\mathrm{diam}(K)=(1/\alpha)\mathrm{diam}(K)$ with
$\mathrm{diam}(K)<\infty$ and $1/\alpha\ne 1$, forcing
$\mathrm{diam}(K)=0$.
\end{proof}

\begin{remark}\label{rem:entropy-vacuous}
Proposition~\ref{prop:no-compact-invariant} shows that the standard
topological-entropy machinery (Bowen spanning sets on a compact
invariant set) cannot detect any non-trivial dynamics of
$\psi_\alpha$ on $(\C,d_\C^s)$ directly: the only candidate compact
invariant sets are $d_\C^s$-trivial.  This is a consequence of the
non-compactness of $\C$ in $\tau_{d_\C^s}$, and a known limitation
of the topological-entropy framework on non-compact spaces.
Alternative invariants suited to non-compact phase spaces (e.g.,
the volume-growth entropy of~\cite{walters1982}, or asymmetric
analogues) lie outside the scope of this paper and are an open
direction.
\end{remark}

\begin{example}[Compact invariant sets are trivial for $\alpha=2$]\label{ex:entropy-alpha2}
For $\alpha=2$, Proposition~\ref{prop:no-compact-invariant} shows
that any compact $\psi_2$-invariant $K\subseteq(\C,d_\C^s)$
satisfies $d_\C^s(f,g)=0$ for all $f,g\in K$.  In particular,
every two-point set $\{f,g\}$ with $d_\C^s(f,g)>0$ (such as
$f(n)=n$, $g(n)=2n$, which have $d_\C^s(f,g)\approx 0.347>0$) is
\emph{not} $\psi_2$-invariant: $\psi_2$ moves it out of itself.
The entropy of $\psi_2$ on any genuinely compact invariant set is
therefore~$0$; non-trivial dynamical complexity of $\psi_2$
manifests only on non-compact invariant sets.
\end{example}

\begin{example}[Singleton invariant sets]\label{ex:entropy-zero}
Let $K=\{f\}$ for any $f\in\C$.  Then $\psi_\alpha(K)=\{\alpha f\}\ne K$
unless $\alpha=1$, so a non-trivial singleton is not
$\psi_\alpha$-invariant for $\alpha\ne 1$.  This is consistent with
Proposition~\ref{prop:no-compact-invariant}: the only compact
invariant sets are $d_\C^s$-trivial, i.e., $d_\C^s$-equivalence
classes, which may be singletons only when the orbit is constant ---
impossible for $\alpha\ne 1$ on distinct points.
\end{example}

\begin{remark}[Entropy on non-compact sets]\label{rem:entropy-gap}
Proposition~\ref{prop:no-compact-invariant} shows that compact
$\psi_\alpha$-invariant sets carry zero entropy.  The natural
question is whether a meaningful entropy invariant can be defined
for the non-compact dynamics of $\psi_\alpha$ on all of $\C$.  On
metric spaces, the variational principle equates topological entropy
with the supremum of measure-theoretic entropies; an analogous
quasi-metric version, if it could be established, might yield a
non-trivial entropy for $\psi_\alpha|_\C$.
\end{remark}

An algorithm for numerically estimating the entropy via spanning
sets is given in
\href{\repourl/blob/main/code/python/entropy_estimate.py}{\texttt{entropy\_estimate.py}}.

% ═══════════════════════════════════════════════════════════════
\section{Conclusion}
\label{sec:conclusion}
% ═══════════════════════════════════════════════════════════════

This paper develops the theory of expansive homeomorphisms on
the complexity quasi-metric space introduced by Schellekens.  The
main results are:

\begin{enumerate}[leftmargin=2em]
  \item \textbf{Expansiveness characterisation}
    (Theorem~\ref{thm:main-scaling}): the scaling map $\psi_\alpha$
    is expansive on $(\C,d_\C)$ if and only if $\alpha\neq 1$.
  \item \textbf{Stable sets as complexity classes}
    (Theorem~\ref{thm:stable-sets}): the $\delta$-stable sets of
    $\psi_\alpha$ coincide with neighbourhoods in $d_\C$ and contain
    all functions that are asymptotically at least as slow.
  \item \textbf{Exact contraction rate} (Corollary~\ref{cor:exact-contraction}):
    distances contract exactly by factor $1/\alpha$ per forward
    iterate, with constant $C=1$.
  \item \textbf{Hierarchy as orbit separation}
    (Theorem~\ref{thm:hierarchy}): the time hierarchy theorem of
    Hartmanis and Stearns corresponds to orbit separation in the
    symmetrized quasi-metric.
\end{enumerate}

Together, these results show that the complexity quasi-metric
space is a workable setting for applying dynamical-systems methods
to computational complexity.

\medskip
The following table summarises the correspondence between dynamical
and complexity-theoretic concepts established in this paper.

\begin{table}[H]
\centering
\renewcommand{\arraystretch}{1.3}
\begin{tabular}{lll}
\hline
\textbf{Dynamical concept} & \textbf{Complexity interpretation} & \textbf{Reference} \\
\hline
Quasi-metric $d_\C(f,g)=0$ & $f$ at least as fast as $g$ & Thm~\ref{thm:dc-props}(ii)\\
Scaling map $\psi_\alpha$ & Uniform speed change by $\alpha$ & Def~\ref{def:scaling}\\
Expansiveness ($\alpha\neq 1$) & Orbits eventually separate & Thm~\ref{thm:main-scaling}\\
$\delta$-stable set & Complexity class neighbourhood & Thm~\ref{thm:stable-sets}\\
Unstable set & All pointwise-faster functions & Thm~\ref{thm:unstable-sets}\\
Exact contraction rate & $d_\C$ decays as $(1/\alpha)^n$ & Cor~\ref{cor:exact-contraction}\\
Backward expansion & $d_\C$ grows as $\alpha^n$ & Cor~\ref{cor:backward-expansion}\\
Orbit separation & Time hierarchy gap & Thm~\ref{thm:hierarchy}\\
No non-trivial compact invariant sets & Entropy of $\psi_\alpha|_K$ is zero & Prop~\ref{prop:no-compact-invariant}\\
\hline
\end{tabular}
\caption{Dictionary between dynamical systems and complexity theory.}
\label{tab:dictionary}
\end{table}

A list of open problems and directions for future work --- including
non-linear and composition-based variants of $\psi_\alpha$, weighted
complexity quasi-metrics, the shadowing property, and the
relationship to entropy on non-compact phase spaces --- is
maintained alongside the paper in
\href{\repourl/blob/main/OPEN-PROBLEMS.md}{\texttt{OPEN-PROBLEMS.md}}
in the companion repository.

\subsection*{Acknowledgements}

The author thanks AIRINA Labs for institutional support.

% ═══════════════════════════════════════════════════════════════

\end{document}